\documentclass[11pt, svgnames]{article}

\usepackage[margin=1in]{geometry}
\usepackage{titling}
\usepackage{kpfonts}
\usepackage{dsfont}
\usepackage{icomma}
\usepackage[demo]{graphicx}
\usepackage{caption}
\usepackage{subcaption}
\usepackage{amsmath, amsthm, amssymb}
\usepackage{physics}
\usepackage{comment}
\usepackage{graphicx}
\usepackage[dvipsnames]{xcolor}
\usepackage{dcolumn}
\usepackage{dsfont}
\usepackage{bm}
\usepackage{bbm}
\usepackage{hyperref}
\usepackage{thmtools} 
\usepackage{thm-restate}
\usepackage{tikz}
\usepackage{bm}
\usetikzlibrary{patterns}
\usetikzlibrary{decorations.pathreplacing}
\usetikzlibrary{calc}
\usepackage{overpic}
\usepackage{mathtools}
\newcommand{\bigO}{\mathcal{O}}
\usepackage{xcolor}
\usepackage{soul}
\usepackage{titling}

\newcommand{\red}[1]{\textcolor{red}{#1}}

\theoremstyle{definition}
\newtheorem{lemma}{Lemma}
\newtheorem{proposition}{Proposition}
\newtheorem{theorem}{Theorem}
\newtheorem{corollary}{Corollary}

\newtheorem{definition}{Definition}

\newtheorem{fact}{Fact}
\newtheorem*{problem}{Problem}

\newcommand{\onenorm}[1]{\left\| {#1} \right\|_1}
\newcommand{\inftynorm}[1]{\left\| {#1} \right\|_{\infty}}
\newcommand{\id}{\mathbbm{1}}

\newcommand{\density}{\mathcal{D}(\mathcal{H})}

\newcommand{\cD}{\mathcal{D}}
\newcommand{\cB}{\mathcal{B}}
\newcommand{\cH}{\mathcal{H}}

\newcommand{\cP}{\mathcal{P}}
\newcommand{\zero}{\dyad{0^n}{0^n}}
\newcommand{\ketzero}{\ket{0^n}}
\newcommand{\ketexcited}{\ket{\psi^{\epsilon}}}

\newcommand{\bounded}{\mathcal{B}(\mathcal{H})}
\newcommand{\params}{(\delta,r,t)}
\newcommand{\triple}{(P,\,\,\Pi, \,\rho_{\Lambda})}
\newcommand{\purdist}{D_{p}}
\newcommand{\fid}{F}
\newcommand{\purifieddistance}{D_{p}}
\newcommand{\locProj}[1]{\frac{\Pi_{R}^{0} {#1} \Pi_{R}^{0}}{\text{Tr}\left[\Pi_{R}^{0} {#1} \Pi_{R}^{0}\right]}}

\usepackage{tcolorbox}
\tcbuselibrary{theorems}

\newtcbtheorem{mytheo}{}%
{colback=blue!5,colframe=blue!35!black,fonttitle=\bfseries}{th}

\usepackage[style=numeric-comp,sorting=none]{biblatex}
\usepackage{hyperref}
\usepackage{xurl}
\hypersetup{breaklinks=true}
\usepackage{hyperref}
\hypersetup{
    colorlinks=true,
    linkcolor=blue,
    citecolor=blue,
    urlcolor=blue
                    }
\makeindex

\title{Circuit depth versus energy in topologically ordered systems}

\date{}
\author{Arkin Tikku\thanks{Centre for Engineered Quantum Systems,
University of Sydney, Sydney, New South Wales 2006, Australia. \href{mailto:atik8650@uni.sydney.edu.au}{atik8650@uni.sydney.edu.au}}  \and Isaac H. Kim\thanks{Department of Computer Science, University of California, Davis, CA 95616, USA.\newline \href{mailto:ikekim@ucdavis.edu}{ikekim@ucdavis.edu}}}

\begin{document}
\maketitle

\begin{abstract}
We prove a nontrivial circuit-depth lower bound for preparing a low-energy state of a locally interacting quantum many-body system in two dimensions, assuming the circuit is geometrically local. For preparing any state which has an energy density of at most $\epsilon$ with respect to Kitaev's toric code Hamiltonian on a two dimensional lattice $\Lambda$, we prove a lower bound of $\Omega\left(\min\left(1/\epsilon^{\frac{1-\alpha}{2}}, \sqrt{\abs{\Lambda}}\right)\right)$ for any $\alpha >0$. We discuss two implications. First, our bound implies that the lowest energy density obtainable from a large class of existing variational circuits (e.g., Hamiltonian variational ansatz) cannot, in general, decay exponentially with the circuit depth. Second, if long-range entanglement is present in the ground state, this can lead to a nontrivial circuit-depth lower bound even at nonzero energy density. Unlike previous approaches to prove circuit-depth lower bounds for preparing low energy states, our proof technique does not rely on the ground state to be degenerate. 
\end{abstract}


\section{Introduction}
\label{sec:introduction}

Understanding the properties of ground- and excited states of local Hamiltonians and calculating their energies is pivotal to advancing condensed matter physics and quantum chemistry \cite{McArdle_2020}. Tackling these problems using classical methods is in general challenging due to the exponential nature of the state space. Even on a quantum computer, it is generally not possible to efficiently estimate ground state energies of local Hamiltonians up to polynomial precision, as established by the quantum Cook-Levin theorem \cite{kitaev2002classical}. 
Finding ways to solve this problem efficiently for restricted classes of Hamiltonians is therefore an active area of research \cite{Lin_20}.

Recent experimental advances have moreover spurred the study of this problem using small-scale, noisy quantum devices \cite{Preskill_2018}.
Given the limitations of current hardware, such as short coherence times, very few qubits and limited connectivity, it is therefore important to understand which lattice models of interest allow for resource-efficient preparation in this constrained setting. A reasonable metric for this assessment is circuit depth, since the most common state preparation methods on near-term devices rely on preparing a state using some variational ansatz quantum circuit  \cite{VQE_original}.

The hope is then that there exist interesting models that require only modest circuit depth, to give good approximations to some target state or energy. If, for example, circuits of polynomial depth are required, one may worry that the short coherence times of such devices would prohibit a good quality of approximation.

In fact, for many commonly studied lattice models in 2D \cite{Kitaev2003, Levin_2005}, preparing exact ground states is known to be a task that in fact does  require circuits of depth polynomial in the system size if the gates are nearest-neighbor only \cite{Bravyi_hastings_verstraete, Haah_2016}. One may however hope that preparing \emph{some} low-energy state as opposed to the exact ground state might be significantly easier. 

In this work, we show that preparing such states still requires circuits of depth polynomial in their inverse energy density. Our main technical contribution is to give a rigorous proof technique for circuit lower bounds of low-energy states of Kiteav's toric code with no assumption on the dimension of the ground space.

We do so by generalizing a formalism  introduced in \cite{Haah_2016}. This formalism allows to prove circuit depth lower bounds for states for which one can find non-trivial pairs of operators that act trivially on sufficiently small subsystems. Finding such operators for an arbitrary state is in general hard.

To prove our circuit depth lower bound for low-energy states, we extend this formalism to allow for operators that slightly perturb the state on the subsystem. We show that operators that act trivially on ground states satisfy this criterion when acting on sufficiently small subsystems of low-energy states and use this fact to arrive at our main result. 

Interestingly, this proof technique does not rely on the dimension of the ground space of the model. It can furthermore be applied to low-energy states of other frustration-free, local commuting models in 2D such as the quantum double or Levin-Wen models \cite{Kitaev2003,Levin_2005}. 

Circuit depth lower bounds for low-energy states of local Hamiltonians have meanwhile played an important role in the quest for a quantum analogue of the PCP theorem \cite{PCP_review,anshu_nirkhe, NLTS}. These results established circuit lower bounds for stabilizer Hamiltonians that are not geometrically local in any fixed spatial dimension and considered their preparation assuming quantum circuits with all-to-all connectivity. An additional key assumption of these works is that dimension of the ground space grows with system size. Our work lifts this assumption and gives a clean and rigorous way to prove lower bounds on the depth of geometrically local circuits that prepare low-energy states of frustration-free, local commuting models with arbitrary ground space dimension. We believe it should be possible to generalize our techniques to local Hamiltonians and circuits on arbitrary interaction graphs and leave this for future work.
\\

Our paper is structured as follows: in Section \ref{sec:prelims} we cover notation and some basic definitions. In Section \ref{sec:main_results}, we state our main result and discuss it in a broader context. In Section \ref{sec:local_invisibility} we introduce the notions of exact and approximate local invisibility of operators, which will form the backbone of our circuit lower bound technique. In Section \ref{sec:LMPs} we will introduce a formal aid that we call a local marginal projector. This type of projection will be key to proving lemmas en route to our main result. In Section \ref{sec:twist_product}, we introduce the twist product, which gives us a way to detect long-range entanglement in a state when combined with locally invisible operators. In Section \ref{sec:low_energy_states} we finally discuss how to adapt our proof technique to low-energy states and prove the main result. We conclude with some discussion on potential future applications of our technique in Section \ref{sec:conclusion}.

\section{Preliminaries and notation}
\label{sec:prelims}
 
Here we cover some basic definitions and notation about quantum spin systems on lattices and distance measures used throughout the paper. 

We consider a quantum spin system on a 2D lattice $\Lambda$ with arbitrary boundary conditions. 
For any pair of sites $v,u\in \Lambda$, their distance $d(v,u)$ is the minimal number of edges that connect $v$ and $u$. For $A, B\subset \Lambda$, their distance is defined as
\begin{equation}
    d(A,B) := \min_{u\in A, v\in B} d(u,v).
\end{equation}
For any site $v\in \Lambda$, define a $r$-disk at $v$ as 
\begin{equation}
    v(r) := \{u: d(v,u)\leq r \}.
\end{equation}
More generally, for any set of sites $A\subset \Lambda$, define 
\begin{equation}
    A(r) := \bigcup_{v\in A} v(r)
\end{equation}
as a $r$-thickening of the region $A$. Given a disc $A$ we will  sometimes refer to $A(-D)$ as the $D$-interior of $A$, which defines all points inside of $A$ that are a distance $D$ away from the boundary of $A$, i.e. $A(-D):=\{v\in A| d(v,x)>D, \forall x\in A^c\}$. We will denote the radius of a disc $A$ by $\text{rad}(A)$.

The global Hilbert space is assumed to have a tensor product structure:
\begin{equation}
    \mathcal{H}_{\Lambda} := \bigotimes_{v\in\Lambda} \mathcal{H}_v,
\end{equation}
where $\dim(\mathcal{H}_v) < \infty$. While we will refer to lattice sites as qubits throughout the text, we shall note that all of our results equally apply to qudit systems. We shall define the Hilbert space associated with subsystem $A\subset \Lambda$ as $\mathcal{H}_A := \bigotimes_{v\in A} \mathcal{H}_v$. The supports of the operators and density matrices shall be denoted in their subscripts by capital roman letters, unless this is obvious from the context. For instance, $O_A$ and $\rho_A$ is an operator and a density matrix, respectively, acting on $\mathcal{H}_A$. 
Given a pure state vector $\ket{\psi}_{\Lambda}$ we shall denote its associated density matrix by $\rho_{\Lambda}$ and it's reduced state (marginal) onto region $A\subset \Lambda$ by $\rho_{A}$, i.e. $\rho_{A}=\Tr_{A^c}[\rho_{\Lambda}]$, where $A^{c}:=\Lambda\backslash A$, i.e. it denotes the complement of $A$ in $\Lambda$. Given an operator $O$, we shall denote its support as $\text{supp}(O)$.

The fidelity between two quantum states $\rho, \sigma \in \mathcal{D}(\mathcal{H})$ is denoted by $\fid(\rho, \sigma) = \|\rho^{\frac{1}{2}}\sigma^{\frac{1}{2}} \|_{1}^{2}$, where $\lVert.\rVert_{1}$ denotes the Schatten 1-norm. A particularly convenient distance measure for our purposes will be the \emph{purified distance} \cite{marco_textbook}, defined as
\begin{equation}
    \purifieddistance(\rho, \sigma) := \sqrt{1-\fid(\rho, \sigma)}.
\end{equation}
From the Fuchs-van-de-Graaf inequality \cite{fuchs1998} it follows that the purified distance upper bounds the trace distance, i.e.  
\begin{align}\label{eq:purdist_tracedist_relations}\frac{1}{2}\onenorm{\rho-\sigma}\leq  \purdist{(\rho, \sigma)} .
\end{align}
For the special case of pure states, the purified distance equals the trace distance, i.e. the above inequality is saturated. 
The purified distance satisfies several desirable properties \cite{marco_textbook}. Let $\Phi: \mathcal{B}(\mathcal{H}) \to \mathcal{B}(\mathcal{H}')$ be a completely positive trace-preserving map and $U$ be a unitary. Then
\begin{align}
    &\purifieddistance(\rho, \tau) \leq \purifieddistance(\rho, \sigma) + \purifieddistance(\sigma, \tau), &\text{(Triangle inequality)} \label{eq:triangle} \\
    &\purifieddistance(\Phi(\rho), \Phi(\tau)) \leq \purifieddistance(\rho, \sigma), &\text{(Monotonicity)} \label{eq:monotonicity} \\
    &\purifieddistance(U\rho U^{\dagger}, U\sigma U^{\dagger}) = \purifieddistance(\rho,\sigma ) &\text{(Unitary invariance)} \label{eq:unitary_invariance}
\end{align}

We will consider geometrically local Hamiltonians $H_{\Lambda}$ on a regular lattice that without loss of generality can be written in the form:
\begin{equation}
    H_{\Lambda} = \sum_{X\subset \Lambda} h_X,\label{eq:ham_form}
\end{equation}
where each $h_X$ is a hermitian operator supported on a set of sites contained in a $2\times 2$ square. Let us remark that one can always coarse-grain a locally interacting quantum many-body systems in two dimensions to ensure such a form. We shall choose a normalization condition such that each $h_X$ satisfies $\inftynorm{h_{X}}= 1$, where $\inftynorm{ \cdot }$ is the operator norm. We shall also assume that the ground state energy is chosen to be $0$; this can be always achieved by shifting and reweighting the terms of the Hamiltonian. 
For our purposes, it will suffice to restrict ourselves to  Hamiltonians $H_{\Lambda}$ for which the global ground state is also a ground state of all the local terms, i.e. if  $\Pi^{0}_{\Lambda}$ is the projection onto the ground space of $H_{\Lambda}$ and  $\Pi^{0}_{X}$ is the projection onto the ground space of a local term $h_{X}$, then  $\Pi^{0}_{X}\Pi^{0}_{\Lambda}=0$.
This property is known as \textit{frustration-freeness}. 

Apart from ground states, we will also be interested in low-energy states $\rho^{\epsilon}$ of a Hamiltonian $H$, where $\Tr[H\rho^{\epsilon}]\leq \epsilon\inftynorm{H}$ for some small $\epsilon\in [0,1]$, which we will refer to as the \textit{energy density}. For lattice Hamiltonians $H_{\Lambda}$ we will assume that $\inftynorm{H_{\Lambda}}=\bigO(\abs{\Lambda})$. 

\begin{figure}[ht]
    \centering
    
\begin{tikzpicture}[scale=0.6]
\draw[step=2, black, thick] (-1,-1) grid (11,11); 
\foreach \x in {0,1,...,5}{
    \foreach \y in {0,1,...,5}{
        \node[draw,circle,fill,minimum size=0.5] at (2*\x,2*\y) {};
      }
    }
\foreach \x in {0,...,5}{
        \node[draw,circle,minimum size= 1] at (2*\x,2*0) {};}

\foreach \y in {1,...,5}{
        \node[draw,circle,minimum size= 1] at (2*5,2*\y) {};

      } 

\foreach \x in {0,4,8}{
    \foreach \y in {0,4,8}{
        \fill[] (\x +0.15,0 + \y +0.15) rectangle (2+\x-0.15,2+\y-0.15);          
      }
    }
\fill[] (2 +0.15,2 +0.15) rectangle (4-0.15,4-0.15);
\fill[] (2 +0.15,6 +0.1) rectangle (4-0.15,8-0.15);
\fill[] (6 +0.15,6 +0.15) rectangle (8-0.15,8-0.15);
\fill[] (6 +0.15,2 +0.15) rectangle (8-0.15,4-0.15); 
        
    \fill[] (-1,-1) rectangle (0-0.15,0-0.15);
    \fill[] (2 +0.15, 0-0.15) rectangle (4-0.15,-1);
    \fill[] (2 +0.15, 0-0.15) rectangle (4-0.15,-1);
    \fill[] (6 +0.15, 0-0.15) rectangle (8-0.15,-1);
    \fill[] (10+0.15,-1) rectangle (11,0-0.15);
    
    \fill[] (-1,10 +0.15) rectangle (0-0.15,11);
    \fill[] (2 +0.15,10 +0.15) rectangle (4-0.15,11);
    \fill[] (6 +0.15,10 +0.15) rectangle (8-0.15,11);        \fill[] (10 +0.15, 10 +0.15) rectangle (11,11);
    \fill[] (10+0.15,-1) rectangle (11,0-0.15);
    
    \fill[] (-1,2 +0.15) rectangle (0-0.15, 4 -0.15);
    \fill[] (-1,6 +0.15) rectangle (0-0.15, 8 -0.15);
    
    \fill[] (10 +0.15,2 +0.15) rectangle (11, 4 -0.15);
    \fill[] (10 +0.15,6 +0.15) rectangle (11, 8 -0.15);
    \node[text=black] at (3, 1) {\Large $\bm{s^{X}_{i}}$};
    \node[text=white] at (5, 5) {\Large $\bm{s^{Z}_{i}}$};
\end{tikzpicture}
    \caption{A $L=6$ instance of the square toric code lattice with qubits placed on vertices and periodic boundary conditions. Black faces denote $Z$-type stabilizers and white faces denote $X$-type stablizers.} 
    \label{fig:toric_lattice}
\end{figure}
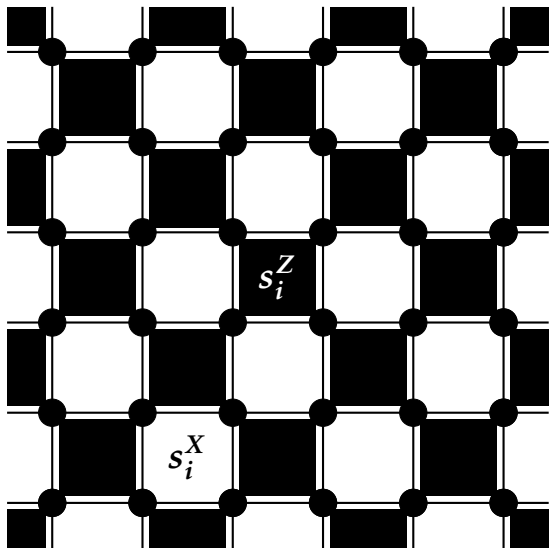

While our results apply to more general classes of models such as Kitaev's quantum double \cite{Kitaev2003} or the Levin-Wen model \cite{Levin_2005}, we focus on Kitaev's toric code~\cite{Kitaev2003} for concreteness. We will sometimes refer to the ground spaces of these models containing states with \textit{long-range entanglement} or \textit{topological order}, by which we mean that they cannot be prepared using geometrically local circuits of constant depth \cite{Haah_2016, Bravyi_hastings_verstraete, Koenig_Pastawski}. This notion can be made more mathematically precise, as we will do in Section \ref{sec:twist_product}.

The toric code Hamiltonian on a lattice $\Lambda$ can be defined in terms of the generators of the stabilizer group $S^{\text{TC}}_{\Lambda}$ of the toric code (see~\cite{GottesmanThesis, Terhal_2015_review} for reviews on stabilizer codes and the toric code). Specifically, one may write it as the sum of projections onto the $(-1)$-eigenspaces of the generators, so that the toric code forms its ground space i.e. 
\begin{align}\label{eq:toric_hamiltonian}
H_{\Lambda}=\sum_{i}\frac{\id-s^{X}_{i}}{2} + \sum_{i}\frac{\id-s^{Z}_{i}}{2}, 
\end{align}
where $\{ s_{i}^{X}\}$ and $\{s_{i}^{Z}\}$ denote the $X$-type and $Z$-type generators respectively, which are geometrically local (see Figure \ref{fig:toric_lattice}). Sometimes we will be interested in considering the \textit{restriction} of the Hamiltonian $H_{\Lambda}$ defined on the full lattice to a region $R\subset\Lambda$. Whe shall denote this restricted Hamiltonian as $H_{R}$, which can be obtained from $H_{\Lambda}$ by simply omitting all the terms in Equation (\ref{eq:toric_hamiltonian}), whose support is not contained in $R$. Sometimes we will be interested in the projection onto the ground space of $H_R$, which we shall denote as $\Pi_R^{0}$,

\begin{equation}\label{eq:loc_proj}
    \Pi_R^{0}:= \prod_{i} \frac{\id-s^{Z}_{i}}{2}\prod_{i} \frac{\id-s^{X}_{i}}{2},
\end{equation}
where once again we only consider $X$-type and $Z$-type generators whose support is contained in $R$.

\section{Main Result}
\label{sec:main_results} 

\noindent Our work considers the following question: 

\begin{problem}(Low-energy state preparation)
Let $\Lambda$ be a $2D$ lattice and let $H_{\Lambda}$ be a frustration-free, commuting Hamiltonian whose ground space is topologically ordered (this notion will be made mathematically precise later in Section \ref{sec:twist_product}). Let $\rho^{\epsilon}_{\Lambda}$ be a state of global energy density at most $ \epsilon$, i.e. $\Tr[H_{\Lambda}\rho^{\epsilon}_{\Lambda}]\leq|\Lambda| \epsilon$ and let $\ketexcited_{\Lambda'}= U\ket{0}^{\otimes \abs{\Lambda'}}$ denote its purification on an extended lattice $\Lambda'=\Lambda \cup E $, where $E$ denotes a register of purifying ancillas. What is the smallest depth of any quantum circuit implementing $U$ only using geometrically local gates in 2D ?  
\end{problem}

Our main result gives an answer to this question for the case in which the Hamiltonian is given by the toric code Hamiltonian with arbitrary ground space degeneracy and arbitrary boundary conditions:

\begin{restatable}{theorem}{main}
Let $\rho^{\epsilon}_{\Lambda}$ be a state with energy density at most $\epsilon$ with respect to the toric code Hamiltonian $H_{\Lambda}$ on a 2D  lattice $\Lambda$ with arbitrary boundary conditions.
Then any local quantum circuit that prepares $\rho^{\epsilon}_{\Lambda}$ from a product state must be of depth
$ D=\Omega\left(\min\left(1/\epsilon^{\frac{1-\alpha}{2}}, \sqrt{\abs{\Lambda}}\right)\right)$ for any $\alpha>0$.
\end{restatable}

This is an unexpected result, for the reasons discussed below. It is well-known that the toric code Hamiltonian at finite temperature has no topological order, as does a topologically trivial Hamiltonian. More precisely, such states can be created by applying a finite number of layers of local quantum channels~\cite{hastings_2011}. Therefore, when it comes down to the properties of these systems at finite energy density, it is natural to expect those states to have similar properties. In particular, it would be natural to conjecture that the complexity of preparing a state with energy density $\epsilon$ is, up to a polylogarithmic factor, independent of whether the underlying Hamiltonian is topologically ordered or not at zero temperature. 

If one were to accept such a conjecture, one would be inclined to conclude that the circuit depth for preparing the low-energy states of a topologically trivial Hamiltonian would be also polynomial in $1/\epsilon$. However, this seems unlikely to be true. For instance, we can consider the trivial Hamiltonian $H=\sum_{i} Z_i$, for which even the exact ground state can be prepared in $O(1)$ depth. Even if a small perturbation is added to such a Hamiltonian, the new ground state can be mapped to the trivial state by a locality-preserving unitary known as quasi-adiabatic continuation~\cite{Hastings_quasi_2005}. Since there is an approach to simulate locally interacting time-dependent Hamiltonian that can approximate local properties up to an error $\delta$, with a circuit depth that scales polylogarithmically with $1/\delta$~\cite{microsoft_2019,Alhambra2021}, it is natural to expect that the quasi-adiabatic continuation can be also simulated with a similar complexity. Thus, our result suggests that, despite being both ``trivial,'' there is a subtle sense in which the finite-temperature state of the toric code Hamiltonian and the trivial Hamiltonian are different. 

On the conceptual side, our primary contribution lies in pointing out this subtle difference by proving a new circuit depth lower bound. While the purported polylogarithmic circuit-depth upper bound for general topologically trivial Hamiltonian has not been proven rigorously yet, we believe this is a reasonable conjecture in light of our plausibility argument above. 

Our second conceptual contribution can be viewed as establishing a fundamental limitation in approximating local properties of ground states using geometrically local circuits. If such a circuit prepares a state that can approximate local expectation values of the ground state up to an error $\epsilon$, then this circuit by definition prepares a state with energy density $O(\epsilon)$. Our lower bound shows that, in the worst case, the circuit depth required to achieve such an error must thus grow polynomially with $1/\epsilon$. 

We note that this limitation is not just of theoretical interest. There is a vast literature of numerical studies of the variational quantum eigensolver (VQE)~\cite{VQE_original}, and many focus on how the error with respect to the ground state energy \emph{per site} decays with the circuit depth. In particular, we note that there are results in the literature which report a decay the of the energy that is exponential in the circuit depth of the variational ansatz ~\cite{Montanaro_Bosse,Montanaro_fermi_hubbard}. Our result shows rigorously that such a scaling can not be achieved in general, if the ground space of the cost Hamiltonian is topologically ordered. Understanding the performance of VQE is in general a difficult task, and our work provides one approach to elucidate its limitations.

On the technical front, our main contribution lies in a novel proof technique for proving robust circuit depth lower bound. Most notably, our lower bound does not rely on the ground state degeneracy and therefore removes one of the key assumptions in recent works on circuit lower bounds for low-energy states of constant-rate stabilizer codes \cite{anshu_nirkhe, NLTS}. Moreover, our proof technique relies only on the existence of certain unitary ``loop operators'' of constant depth that stabilize the ground space. As such, the technique is applicable to a wider class of models such as Kitaev's quantum double~\cite{Kitaev2003} and Levin-Wen string-net models~\cite{Levin_2005}.


We note that our proof strategy is inspired  by the work of Haah \cite{Haah_2016}, which gave a rigorous circuit-depth lower bound for the ground state of the toric code on a sphere, which has a unique ground state. We provide a more robust version of this technique, which is applicable to general mixed states, not just pure states or ground states.

\subsection{Relation to prior work}
Our work is related to Ref.~\cite{anshu_nirkhe}, which proves a $\Omega(\log(1/\epsilon))$ circuit depth lower bound for the preparation of low-energy states of high-rate stabilizer Hamiltonians with $O(1)$-local but geometrically non-local circuits. However, to the best of our knowledge, it is unclear how to prove our main result using their approach. Therefore, we view our work to be complementary to Ref.~\cite{anshu_nirkhe}; we make a stronger assumption on the form of the circuit and derive a stronger lower bound.

\section{Exact and approximate local invisibility}
\label{sec:local_invisibility}

In this Section, we introduce the concept of \emph{approximately locally invisible operators}, which will play a crucial role in the proof of our main result. Local invisibility is a notion first introduced by Haah~\cite{Haah_2016}. To give some intuition for the concept of approximate local invisibility, we first restate Haah's original definition of exact local invisibility.

\begin{definition}(Local invisibility)\label{def:local_invisibility}
Let $A\subset{\Lambda}$ be a disc of radius at most $r$ and let $B:=A(t')\subset{\Lambda}$ for some $t'\geq t \geq 0$. An operator $P\in \mathcal{B} (\mathcal{H}_{\Lambda})$ is $(A,B)$-locally invisible with respect to a state $\rho_{\Lambda}$ if for a disc $A\subset B\subset\Lambda$ and any state $\rho_{\Lambda}'$ such that $\rho_{B}' = \rho_B$, we have
\begin{equation}
    \rho_A = \frac{\text{Tr}_{A^c} \left[ P \rho_{\Lambda}'P^{\dagger} \right]}{\text{Tr}\left[ P\rho_{\Lambda}' P^{\dagger}\right]}.
\end{equation}
An operator is then $(r,t)$-locally invisible with respect to a state $\rho_{\Lambda}$ if it is $(A,B)$-locally invisible with respect to $\Lambda$ with $A$ and $B$ specified as above. 
\end{definition}

One class of operators that satisfy this condition for a state $\rho_{\Lambda}$ is any operator $P$ that stabilizes a state $\rho_{\Lambda}$, i.e. $P\rho_{\Lambda}P^{\dagger}=\rho_{\Lambda}$. If this operator is a unitary of depth $T$, then we may even fully characterize its parameters.

\begin{fact}\label{fact:exact_invisibility}
Let $\rho_{\Lambda}\in \cD(\cH_{\Lambda})$ and let $P\in \cB(\cH_{\Lambda})$ be a unitary operator such that $P\rho_{\Lambda}P^{\dagger}=\rho_{\Lambda}$.
Suppose that $P$ is of circuit depth at most $T$ with respect to nearest-neighbor gates.  Then $P$ is $(r,t)$-locally invisible with respect to $\rho_{\Lambda}$ for $r=\bigO\left(\sqrt{\Lambda}\right)$ and $t=\Omega\left(T\right)$.
\end{fact}
\begin{proof}
The proof follows immediately from the assumption that $P$ stabilizes $\rho_{\Lambda}$ and requires one to choose $t\geq T$ in order for the definition to hold. 
\end{proof}

Let us give a concrete example of such an operator. 
Define $\rho_{\Lambda}^{0}=\locProj{\rho_{\Lambda}}$, for some $\rho_{\Lambda}$, where $\Pi_R^{0}$ is the projection onto the local toric code ground space as defined in Eq.~\eqref{eq:loc_proj} and assume that $\text{Tr}[\Pi_R^{0} \rho_{\Lambda} \Pi_R^{0}]>0$. Consider a toric code stabilizer $s^{Z}\in S_R^{\text{TC}}$ which is a tensor product of $Z$ that forms a ``loop'' inside $R$; see Fig. \ref{fig:loop_op}. One can show that this loop is $(A, A(2))$-locally invisible with respect to $\rho^{0}_{\Lambda}$ for any disc $A\subset \Lambda$, because all the toric code stabilizers supported strictly on $R$  stabilize $\rho^{0}_{\Lambda}$. Note how the support of $s^{Z}_{R}$ is crucial, i.e. any $Z$-type loop operator with support overlapping with $R^{c}$, is in general not locally invisible with respect to $\rho^{0}_{\Lambda}$.

\begin{figure}[ht]
\centering
 \begin{tikzpicture}[scale=1.25]
\draw[ultra thick] (0,0) rectangle (5,5);
\draw[ultra thick] (0.8,0.8) rectangle (4.2,4.2);
\node at (4.5,0.5) {$\Lambda$};
\node at (3.8,1.2) {$R$};
\draw[thick, fill=cyan, fill opacity=0.3] (1.5,1.5) circle (0.53cm);
\draw[thick,fill=yellow,] (1.5,1.5) circle (0.3cm);
\node at (1.5,1.5) {$A'$};
\node at (1.65, 1.1) {$B'$};
\draw[thick,fill=cyan, fill opacity=0.3] (3.1,3.1) circle (0.5cm);
\draw[thick,fill=yellow ] (3.1,3.1) circle (0.3cm);
\node at (3.1,3.1) {$A$};
\node at (3.25, 2.72) {$B$};
\draw[ultra thick, red] (2.5,2.5) circle (1 cm);
\node at (2.5,3.8) {$\red{s^{Z}}$};

\end{tikzpicture}
\caption{A $Z$-type loop operator $s^{Z}\in \mathcal{S}_{R}^\text{TC}$ that stabilizes the ground space  of $H_{R}$. Its annular support is contained in $R\subset \Lambda$ and it is locally invisible with respect to any state in the ground space of $H_R$. Locally invisible regions $A$ and $A'$ are shown with $B=A(t)$ for some $t=\bigO(1)$ and likewise for $B'$.}
\label{fig:loop_op}
\end{figure}
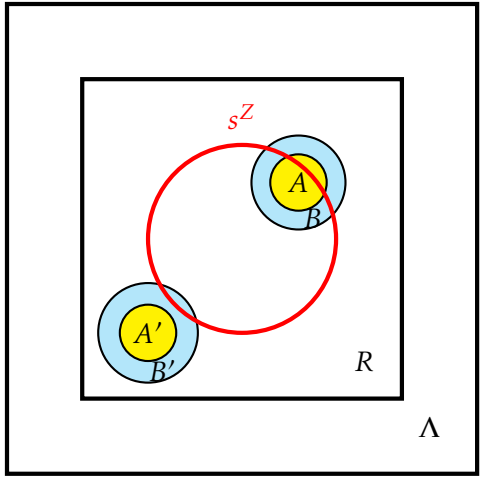


For the state $\rho_{\Lambda}$ prior to the projection, it is generally not clear how to find a nontrivial operator which is locally invisible. However in some cases, as we will see later, one may be able to find an operator that approximately satisfies local invisibility.

\begin{definition}(Approximate local invisibility)
\label{def:approx_loc_invisible}
Let $\delta\in [0,1]$ and let $A\subset B \subset \Lambda$ be as given in Definition \ref{def:local_invisibility} . An operator $P\in \mathcal{B}(\mathcal{H}_{\Lambda})$ is $\params$-locally invisible with respect to a state $\rho_{\Lambda}$ if for any state $\rho_{\Lambda}'$ such that $\rho_B = \rho_B'$, we have
\begin{equation}
    D_{p} \left(\rho_A, \frac{\text{Tr}_{A^c} \left[ P \rho_{\Lambda}'P^{\dagger} \right]}{\text{Tr}\left[ P\rho_{\Lambda}' P^{\dagger}\right]} \right) \leq \delta.
\end{equation}
\end{definition}

In Definition~\ref{def:approx_loc_invisible}, we used the purified distance as our distance measure of choice; see Section \ref{sec:prelims} for details. This is solely for convenience and can directly be translated into a statement using the trace distance via Eq.~\eqref{eq:purdist_tracedist_relations}.

The relationship between exact and approximate local invisibility becomes interesting when considering unitary operators and states that are sufficiently close to each other in purified distance. Namely, if a unitary operator is exactly locally invisible with respect to some state, then it is approximately locally invisible with respect to any sufficiently small perturbation of the state. 

\begin{lemma}\label{lem:robustness_local_invisibility}
Let $\rho_{\Lambda}$ and $\sigma_{\Lambda}$ be two quantum states such that  $\purifieddistance{\left(\rho_{\Lambda},\sigma_{\Lambda}\right)}\leq \epsilon$. Then if $P$ is a $(r,t)$-locally invisible unitary with respect to $\rho_{\Lambda}$ for any disc $A$ of radius at most $r$ and any disc  $B:=A(t)$ for $t\geq \text{depth}(P)$, then $P$ is $(2\epsilon, r,t)$-approximately locally invisible with respect to  $\sigma_{\Lambda}$.
\end{lemma}

\begin{proof}
Our goal here is to show that for any state $\sigma'_{\Lambda}$ s.t. $\sigma'_{B}=\sigma_{B}$, we have that \begin{align}
\purifieddistance{\left(\sigma_{A},\Tr_{A^c}\left[P\sigma'_{\Lambda}P^{\dagger}\right]\right)}\leq 2\epsilon.
\end{align} This immediately follows from triangle inequality and our assumptions. We have 
\begin{align}
\purifieddistance{\left(\sigma_{A},\Tr_{A^c}\left[P\sigma'_{\Lambda}P^{\dagger}\right]\right)}&\leq \purifieddistance{\left(\sigma_{A}, \rho_{A}\right)}+ \purifieddistance{\left(\rho_{A},\Tr_{A^c}\left[P\rho'_{\Lambda}P^{\dagger}\right]\right)}\nonumber\\
&+ \purifieddistance{\left(\Tr_{A^c}\left[P\rho'_{\Lambda}P^{\dagger}\right],\Tr_{A^c}\left[P\sigma'_{\Lambda}P^{\dagger}\right]\right)} \label{eq:robust_ineq}\\
&\leq \epsilon + 0 + \epsilon = 2\epsilon,
\end{align}
where the bound on the first term follows from our closeness assumption and monotonicity of the purified distance and the second term is equal to zero since $P$ is exactly locally invisible with respect to $\rho_{\Lambda}$.

We bound the third term on the RHS of line \ref{eq:robust_ineq} as follows: Suppose $T:=\text{depth}(P)$. Then since $\sigma'_{B}=\sigma_{B}$, the lightcone of the depth $T$ unitary $P$ implies that $\Tr_{B(-T)^c}\left[P\sigma'_{\Lambda}P^{\dagger}\right]=\Tr_{B(-T)^c}\left[P\sigma_{\Lambda}P^{\dagger}\right]$, where $B(-T)$ is the $T$-interior of $B$ as defined in \ref{sec:prelims}. Since $t\geq T$ we have $A\subseteq B(-T)$ from which it follows that $\Tr_{A^c}\left[P\sigma'_{\Lambda}P^{\dagger}\right]=\Tr_{A^c}\left[P\sigma_{\Lambda}P^{\dagger}\right]$. As $P$ is defined to be exactly locally invisible, i.e. $\Tr_{A^c}\left[P\rho'_{\Lambda}P^{\dagger}\right]=\Tr_{A^c}\left[P\rho_{\Lambda}P^{\dagger}\right]$  we therefore have that 
\begin{align}
\purifieddistance{\left(\Tr_{A^c}\left[P\rho'_{\Lambda}P^{\dagger}\right],\Tr_{A^c}\left[P\sigma'_{\Lambda}P^{\dagger}\right]\right)}=\purifieddistance{\left(\Tr_{A^c}\left[P\rho_{\Lambda}P^{\dagger}\right],\Tr_{A^c}\left[P\sigma_{\Lambda}P^{\dagger}\right]\right)} \leq \epsilon,
\end{align}
where the inequality from unitary invariance and monotonicity of the purified distance and our assumption that $\purifieddistance{(\rho_{\Lambda},\sigma_{\Lambda})}\leq \epsilon$.
\end{proof}

Our lemma shows that, assuming that $P$ is a unitary, the notion of local invisibility with respect to a state is in a sense 'robust' to perturbations to the state. 

It is interesting to note that for stabilizer Hamiltonians, the stabilizers themselves yield constant-depth unitaries that are exactly locally invisible with respect to their ground states. Also for more general local commuting projector models, such as the quantum double or string-net models \cite{Kitaev2003, Levin_2005} there exist loop operators that act trivally on any state in the ground space. These are sometimes referred to as Wilson loop operators \cite{Barkeshli_22}. We will use this fact later to find approximately locally invisible operators for our low-energy states.

\section{Local marginal projectors}\label{sec:LMPs}

A particularly useful notion to define in order to prove our main result is the concept of a  \emph{local marginal projector} (LMP for short).
A LMP can be used to fix the marginal of a state $\sigma_{\Lambda}$ on some region $S$ to a desired marginal $\rho_{S}$. The main purpose of this section is to prove properties of  local marginal projectors, since they play an important role en route to our main theorem.

Let us first define local marginal projectors and then explain the key result of this section at a high level.

\begin{definition}(Local marginal projector)
Let $\zero_{S}$ be the projection onto the all-zero state supported on a region $S\subset\Lambda'$ and let $U$ be a local quantum circuit of depth-$D$ that prepares a state $\ket{\psi}_{\Lambda'}$, i.e. $\ket{\psi}_{\Lambda'}=U\ketzero_{\Lambda'}$. Then we say that 
\begin{align}
\Pi_{S(D)}:= U \left( \zero_{S}\otimes\mathbbm{1}_{S^c}\right )U^{\dagger}
\end{align}
is a \textit{local marginal projector} (LMP) on $S(D)$ with respect to the state $\ket{\psi}_{\Lambda'}$.
\end{definition}

Throughout the following sections, we will routinely consider an operator $P$ that is $\params$-locally invisible with respect to a state $\rho_{\Lambda}\in \cD(\cH_{\Lambda})$ and its purification $\ket{\psi}_{\Lambda'}\in \cH_{\Lambda'}$, and LMPs $\Pi_{S(D)}$ with respect to this purification. To signify this relationship between a state $\rho_{\Lambda}$, associated LMPs $\Pi_{S(D)}$ and an operator $P$ succinctly, we will introduce the following definition.

\begin{definition}
Let $\rho_{\Lambda}\in \cD(\cH_{\Lambda})$ be a state with purification $\ket{\psi}_{\Lambda'}$ on an extended lattice $\Lambda'=\Lambda \cup E$, where $E$ is a register of purifying ancillas. Let $P\in \cB(\cH_{\Lambda})$ be a unitary operator of depth at most $t$ that is $\params$-approximately locally invisible with respect to $\rho_{\Lambda}$ and let
$\Pi:=\left\{\Pi_{S(D)}\right\}$ be the set of LMPs $\Pi_{S(D)}$ with respect to $\ket{\psi}_{\Lambda'}$ where $S(D)\subset{\Lambda}$. 
Then we call $\triple$ a $\params$-triple. 
\end{definition}
With slight abuse of notation, we will sometimes refer to an LMP whose support can remain unspecified with $\Pi$, since this will be clear from context. 
We will later also consider states $\rho_{\Lambda}$ for which we can find pairs $P$ and $Q$ of $\params$-locally invisible operators. We will then denote the associated $\params$-triple by $\left(\{P,Q\}, \,\,\Pi,\,
\rho_{\Lambda}\right)$. 
\\
\\
Having introduced these two key concepts, let us now explain the main technical statement that we intend to prove in this section. Roughly speaking, we intend to show that for any $\params$-triple $\triple$, we have that 
\begin{equation}
    P\Pi_{\cP'} \approx \expval{P}{\psi} \Pi_{\cP'},
\end{equation}
where $\cP'$ is some sufficiently large thickening of $p:=\text{supp}(P)$. This will prove to be a useful fact to property of $\params$-triple in order to prove our main theorem. We will make this statement more precise in Proposition \ref{prop2}.  

\subsection{Properties of LMPs}
We will now state some basic facts about LMPs. Note that all the LMPs appearing below will be defined with respect to the same state $\ket{\psi}_{\Lambda}$. For this section we will only consider a lattice $\Lambda$ and no extension $\Lambda'$.

\begin{fact}(Stabilization property)\label{fact:LMP_stab}
$\ket{\psi}_{\Lambda}=\Pi_{S(D)}\ket{\psi}_{\Lambda}$ for any $S(D)\subset \Lambda$. 
\end{fact}

\begin{fact}(Union property)\label{fact:union_prop}
Let $S_{1},S_{2}\subset\Lambda$. Then $\Pi_{S_{1}(D)}\Pi_{S_{2}(D)}=\Pi_{(S_{1}\cup S_{2})(D)}$
\end{fact}


\begin{fact}\label{fact3}
Let $\Lambda=\bigcup_{i}S_{i}$ and let $\Pi_{S_{i}(D)}$ the LMP on $S_{i}(D)$ with respect to the state $\ket{\psi}_{\Lambda}$. Then
\begin{align}
\dyad{\psi}{\psi}_{\Lambda}=\prod_{i}\Pi_{S_{i}(D)}.
\end{align}
\end{fact}

\begin{fact}\label{fact4}
For any $S,S'\subset \Lambda$ we have that $\left[\Pi_{S(D)}, \Pi_{S'(D)}\right ]=0$.
\end{fact}

A local marginal projector $\Pi_{S(D)}$ with respect to a state $\ket{\psi}_{\Lambda}$ locally projects any state $\sigma_{\Lambda}$ onto $\Tr_{S'^{c}}[\dyad{\psi}{\psi}_{\Lambda}]$
for a disc $S'\subset S$ s.t. $S'(D)=S$.
More formally, we have:
\begin{lemma}\label{lemma1}
Let $S\subset \Lambda$ be a disk of radius $r>D$ and let  $\Pi_{S(D)}$ be a LMP with respect to a state $\ket{\psi}_{\Lambda}=U\ketzero_{\Lambda}$. Let  $\sigma_{\Lambda}$ be an arbitrary state such that $\Pi_{S(D)}\sigma_{\Lambda}\Pi_{S(D)}\neq 0$. Then for any concentric disk of $S'\subset S\subset \Lambda$ of radius $r-D$ or less,  
\begin{equation}\label{eq:marginals}
\Tr_{S'^{c}}\left[\dyad{\psi}{\psi}_{\Lambda}\right]=\frac{\Tr_{S'^{C}}\left[\Pi_{S(D)}\sigma_{\Lambda}\Pi_{S(D)}\right] }{\Tr[\Pi_{S(D)}\sigma_{\Lambda}]}.
\end{equation}
\end{lemma}

\begin{proof}
We prove the claim by showing that the expectation value of any $O_{S'} \in \mathcal{B}(\mathcal{H}_{S'})$ is identical for $|\psi\rangle_{\Lambda}$ and any purification of $\sigma_{\Lambda}'= \frac{\Pi_{S(D)}\sigma_{\Lambda}\Pi_{S(D)} }{\Tr[\Pi_{S(D)}\sigma_{\Lambda}]}$, denoted as $|\phi\rangle_{\Lambda'}$.\footnote{Here $\Lambda'$ is the purifying space.} The equivalence of the two expectation values can be established by showing that $U^{\dagger}|\psi\rangle_{\Lambda}$ and $U^{\dagger}\Pi_{S(D)}|\phi\rangle_{\Lambda'}$, upon normalizing the states, have identical marginals over $S$. Note that the support of $U$ and $\Pi_{S(D)}$ is still confined to $\Lambda$. Keeping this in mind, note that
\begin{equation}
\begin{aligned}
U^{\dagger}|\psi\rangle_{\Lambda} &= |0^n\rangle_S \otimes |0^n\rangle_{\Lambda \setminus S}, \\
U^{\dagger}\Pi_{S(D)}|\phi\rangle_{\Lambda'} &= |0^n\rangle_S \otimes |\phi'\rangle_{\Lambda' \setminus S}, 
\end{aligned}
\end{equation}
where $|\phi'\rangle$ is some (possibly unnormalized) state with a positive norm. Since $U$ is a depth-$D$ quantum circuit and $\ket{\phi}_{\Lambda'}$ is a purification of $\sigma'_{\Lambda}$,  our claim follows from the standard ``lightcone'' arguments. 

For the interested readers, we explain the details below. Note that \begin{equation}
U^{\dagger}\Pi_{S(D)}\ket{\phi}_{\Lambda'}=\left(\dyad{0^n}{0^n}_{S}\otimes \mathbbm{1}_{S^c}\right)U^{\dagger}\ket{\phi}_{\Lambda'}= \ketzero_{S}\otimes |\phi'\rangle_{\Lambda'\setminus S},
\end{equation}
where $\ket{\phi'}$ is some nonzero vector. 
Also recall that $U^{\dagger}\ket{\psi}_{\Lambda}=\ketzero_{\Lambda}=\ketzero_{S}\otimes \ketzero_{S^c}$.
\begin{align}\label{eq:24}
\expval{O_{S'}}{\psi}-\frac{\expval{\Pi_{S(D)}O_{S'}\Pi_{S(D)}}{\phi}}{\expval{\Pi_{S(D)}}{\phi}}
&=\expval{U^{\dagger}O_{S'}U}{0^n}-\frac{\bra{\phi'}_{\Lambda'\setminus S}\otimes\bra{0^n}_{S}U^{\dagger}O_{S'}U\ket{\phi'}_{\Lambda' \setminus S}\otimes\ket{0^n}_{S}}{\langle \phi'| \phi' \rangle}\\
&=\Tr\left[U^{\dagger}O_{S'}U\dyad{0^n}{0^n}_{S}\right]-\Tr\left[U^{\dagger}O_{S'}U\dyad{0^n}{0^n}_{S}\right] \\
&=0.\label{eq:26}
\end{align}
Here we used the fact that $U^{\dagger}O_{S'}U$ is supported on $S$ to arrive at the second line.
\end{proof}

\begin{lemma}\label{fact5}
Let $O\in\mathcal{B}(\mathcal{H}_{\Lambda})$ and let $\Pi_{S(D)}$ be a LMP with respect to $\ket{\psi}_{\Lambda}=U\ketzero_{\Lambda}$. Then for any region $S'$, s.t. $\text{supp}(O)\subset{S'}$ and $S=S'(D)$ we have that: 
\begin{align}
\Pi_{S(D)}O\Pi_{S(D)}=\expval{O}{\psi}\Pi_{S(D)}
\end{align}
\end{lemma}
\begin{proof}
The proof follows in analogy to equations (\ref{eq:24}-\ref{eq:26}), but by showing that
\begin{align}
\expval{O}{\psi}-\frac{\bra{\phi}\Pi_{S(D)}O\Pi_{S(D)}\ket{\phi'}}{\bra{\phi}\Pi_{S(D)}\ket{\phi'}}=0
\end{align}
for any $\ket{\phi}_{\Lambda},\ket{\phi'}_{\Lambda}\in \mathcal{H}_{\Lambda}$

\end{proof}

\subsection{LMPs and local invisibility}
In this section, we will develop some technical lemmas on how LMPs can be used to put bounds on the action of approximately locally invisible operators. 

\begin{lemma}\label{lem4}
Let $\triple$ be a $\params$-triple. Let $S$ be any disc of radius at most $r-D$ and let $B:=S(t')$ for $t'\geq t+D$ (see Figure \ref{fig:disc_arrangement}). Then  for any $\sigma_{\Lambda}$, such that $\sigma_{B}=\rho_{B}$, we have 
\begin{align}
\onenorm{\Pi_{S(D)} P \sigma_{\Lambda}P^{\dagger}\Pi_{S(D)}-P\sigma_{\Lambda}P^{\dagger}}\leq 2\delta.
\end{align}
\end{lemma}

\begin{proof}
This claim follows directly from the approximate local invisibility of $P$ and the fact that LMPs with respect to some state $\ket{\psi}_{\Lambda'}$ supported on some region $S$ stabilize any state that looks like the marginal of $\ket{\psi}_{\Lambda'}$ on $S$.
By our local invisibility assumption we have that
\begin{align}
\purdist\left(\rho_{A},\Tr_{A^c}\left[P\sigma_{\Lambda}P^{\dagger}\right]\right)\leq \delta
\end{align}
for any disk $A$ of radius $r$ or less, which is concentric with region $B$. It is immediate that 
\begin{align}
\fid\left(\rho_{A},\Tr_{A^c}\left[P\sigma_{\Lambda}P^{\dagger}\right]\right) \geq 1-\delta^2.
\end{align} 

Let $\Lambda'$ be our extended lattice on which we define purifcations $\ket{\psi}_{\Lambda'}$ and $P\ket{\phi}_{\Lambda'}$ of $\rho_{A}$ and $\Tr_{A^c}\left[P\sigma_{\Lambda}P^{\dagger}\right]$ respectively. 
By Uhlmann's theorem there exists a unitary $U_{A^c}$ acting on the purifiying space $A^c$, such that
$\left|_{\Lambda'}\bra{\psi}U_{A^c}P\ket{\phi}_{\Lambda'}\right|^2=\fid\left(\rho_{A},\Tr_{A^c}\left[P\sigma_{\Lambda}P^{\dagger}\right]\right)$
implying that
\begin{align}
\purdist\left(U_{A^c}\dyad{\psi}{\psi}_{\Lambda'}U_{A^c}^{\dagger},P\dyad{\phi}{\phi}_{\Lambda'}P^{\dagger}\right)\leq \delta
\end{align}

Now let $A=S(D)$, where $\text{rad}(S)\leq r-D$. Then we have that
\begin{align}
\onenorm{U_{S(D)^c}\dyad{\psi}{\psi}_{\Lambda'}U_{S(D)^c}^{\dagger}-P\dyad{\phi}{\phi}_{\Lambda'}P^{\dagger}}\leq \purdist\left(U_{S(D)^c}\dyad{\psi}{\psi}_{\Lambda'}U_{S(D)^c}^{\dagger}, P\dyad{\phi}{\phi}_{\Lambda'}P^{\dagger}\right)\leq \delta
\end{align}

\noindent Since $\onenorm{XYZ}\leq \left\Vert X\right\Vert_{\infty}\onenorm{Y} \left\Vert Z\right\Vert_{\infty}$ for any $X,Y, Z \in \bounded$,
and since $\left\Vert \Pi_{S(D)}\right\Vert_{\infty}=1$, we get
\begin{align}
\onenorm{\Pi_{S(D)}U_{S(D)^c}\dyad{\psi}{\psi}_{\Lambda'}U_{S(D)^c}^{\dagger}\Pi_{S(D)}-\Pi_{S(D)}P\dyad{\phi}{\phi}_{\Lambda'}P^{\dagger}\Pi_{S(D)}}\leq \delta.
\end{align}
Moreover, since $\left[\Pi_{S(D)},U_{S(D)^c}\right]=0$ and $\Pi_{S(D)}\ket{\psi}_{\Lambda'}=\ket{\psi}_{\Lambda'}$, it follows that
\begin{align}
\onenorm{U_{S(D)^c}\dyad{\psi}{\psi}_{\Lambda'}U_{S(D)^c}^{\dagger}-\Pi_{S(D)}P\dyad{\phi}{\phi}_{\Lambda'}P^{\dagger}\Pi_{S(D)}}\leq \delta    
\end{align}
Via triangle inequality it then follows
\begin{align}\label{eq:lem_2_purified}
\onenorm{\Pi_{S(D)} P \dyad{\phi}{\phi}_{\Lambda'}P^{\dagger}\Pi_{S(D)}-P\dyad{\phi}{\phi}_{\Lambda'}P^{\dagger}}\leq 2\delta.
\end{align}
 Noting that  and $\text{supp}(P)\subset \Lambda$ and $S(D)\subset \Lambda$ and using monotonicity of trace distance under partial trace, we thus conclude that $\onenorm{\Pi_{S(D)} P \sigma_{\Lambda}P^{\dagger}\Pi_{S(D)}-P\sigma_{\Lambda}P^{\dagger}}\leq 2\delta$.
\end{proof}
\noindent Lemma \ref{lem4} can then be lifted to a statement in the operator norm about LMPs and approximately locally invisible operators .

\begin{figure}[ht]
\centering
\begin{tikzpicture}[scale=0.8]

\node at (3.9,4.6) {$r-D$};
\node at (4.9, 5.5) {$D$};
\node at (5.6, 6.2) {$t$};
\node at (6.4,6.9) {$D$};

\filldraw [black] (4,4) circle (2pt);
\draw[->,ultra thick] (4.1,4.1)--(4.8,4.8);
\draw[<->,ultra thick] (5.7,5.7)--(6.2,6.2);
\draw[<->,ultra thick] (4.9,4.9)--(5.5,5.5);
\draw[<->,ultra thick] (6.4,6.4)--(7,7);
\draw[ultra thick] (4,4) circle (1.2cm);
\draw[ultra thick] (4,4) circle (2.3cm);
\draw[ultra thick] (4,4) circle (3.3cm);
\draw[ultra thick] (4,4) circle (4.3cm);
\node at (4.7,3.55) {$S$};
\node at (5.7, 3.5) {$S(D)$};
\node at (6.7, 3.5) {$B$};
\node at (7.7,3.5){$S'$};
\end{tikzpicture}
\caption{Arrangements of discs $S, S(D), B$ and $S'$ as described in Proposition \ref{prop1}.}
    \label{fig:disc_arrangement}
\end{figure}
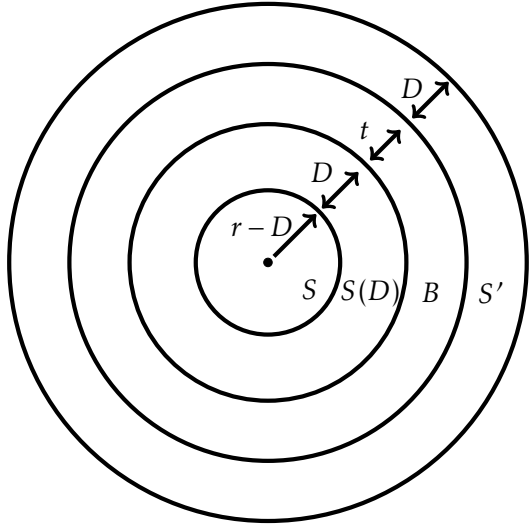

\begin{proposition}\label{prop1}
Let $\triple$ be a $\params$-triple. Let $S$ be any disc of radius at most $r-D$ and let $S':=S(t')$ for $t'\geq t+2D$ (see Fig. \ref{fig:disc_arrangement}). Then 
\begin{align}
\inftynorm{ \Pi_{S(D)}P\Pi_{S'(D)}-P\Pi_{S'(D)}}\leq 2 \delta
\end{align}
\end{proposition}

\begin{proof}
The proof is a simple consequence of combining Lemmas \ref{lemma1} and  \ref{lem4} and changing to the operator norm. Using Lemma \ref{lemma1}, we have that for any disc $B$ of radius at least  $\text{rad}(S)+D+t$, such that $B(D)=S'$:
\begin{align}
    \rho_{B}=\frac{\Tr_{B^{c}}\left[\Pi_{S'(D)}\sigma_{\Lambda}\Pi_{S'(D)}\right]}{\Tr[\Pi_{S'(D)}\sigma_{\Lambda}\Pi_{S'(D)}]},
\end{align}
where $\rho_{B}=\Tr_{B^c}[\rho_{\Lambda}]$.
Invoking Eq.~\eqref{eq:lem_2_purified}, which may be understood as a 'purified version' of Lemma \ref{lem4}, we then have that 
\begin{align}
\onenorm{\Pi_{S(D)}P\Pi_{S'(D)}\dyad{\phi}{\phi}_{\Lambda'}\Pi_{S'(D)}P^{\dagger}\Pi_{S(D)}-P\Pi_{S'(D)}\dyad{\phi}{\phi}_{\Lambda'}\Pi_{S'(D)}P^{\dagger}}\leq 2\delta,
\end{align}
where $\dyad{\phi}{\phi}_{\Lambda'}$is a purification of $\sigma_{\Lambda}$. To bound the complex Euclidean norm between these two states by their trace distance, we shall use the following fact:
\begin{fact}\label{fact6}
Let $\Pi\in \bounded$ be a projector. 
If $\onenorm{\dyad{\psi}{\psi}-\Pi\dyad{\psi}{\psi}\Pi}\leq \epsilon$, then $\left\Vert \ket{\psi}-\Pi\ket{\psi} \right\Vert \leq \epsilon$.
\end{fact}
This can be seen as follows. Since $0\leq \expval{\Pi}{\psi}\leq 1$, we have that 
\begin{align}
\norm{\ket{\psi}-\Pi\ket{\psi}}=\sqrt{1-\expval{\Pi}{\psi}} \leq \sqrt{1-\left(\expval{\Pi}{\psi}\right)^{2}}= \norm{\dyad{\psi}{\psi}-\Pi\dyad{\psi}{\psi}\Pi}_{2}
\end{align}
As $\norm{A}_{2}\leq \norm{A}_{1}$, we arrive at Fact \ref{fact6} . 
\\

This implies that for any state $\ket{\phi}_{\Lambda'}$
\begin{align}
\left\Vert\Pi_{S(D)}P\Pi_{S'(D)}\ket{\phi}_{\Lambda'}-P\Pi_{S'(D)}\ket{\phi}_{\Lambda'}\right\Vert\leq 2\delta,
\end{align}
from which it follows that $\inftynorm{\Pi_{S(D)}P\Pi_{S'(D)}-P\Pi_{S'(D)}}\leq 2\delta$ .
\end{proof}

This statement straightforwardly generalizes to an arbitrary number of discs $S_{i}$ that satisfy the conditions in Proposition \ref{prop1}. 

\begin{corollary}\label{corollary1}
Let $\triple$ be a $\params$-triple. 
Let $\{S_{i}\}$ be any set of $k$ discs of radius at most $r-D$ and let  $\{S'_{i}\}$ be any set of corresponding $k$ discs where $S'_{i}:=S_{i}(t+2D)$.  Then
\begin{align}
\inftynorm{\prod_{i\in[k]} \Pi_{S_{i}(D)}P\prod_{i\in[k]} \Pi_{S_{i}'(D)}-P\prod_{i\in[k]} \Pi_{S_{i}'(D)}}\leq 2k \delta.
\end{align}
\end{corollary} 

\begin{proof}
The proof follows from Proposition \ref{prop1} and repeated application of triangle inequality and Hoelder's inequality.
\end{proof}

Corollary \ref{corollary1} implies that a version of Proposition \ref{prop1} still holds when $\Pi$ is an LMP with support larger than any disc of radius at most $r-D$, since one can divide the support of $\Pi$ into smaller regions of radius at most $r-D$ and then apply Corollary  \ref{corollary1}.

We will now use this technique in order to prove the main proposition of this section.
Suppose $\triple$ is a $\params$- triple. Roughly speaking, if the support of $\Pi$ fully covers some sufficiently large thickening $\cP'$ of the support of the $\params$-locally invisible operator $P$ (see Figure \ref{fig:prop_2_setup}), then the action of $P$ on the LMP is approximately proportional to the expectation value of $P$ with respect to $\ket{\psi}_{\Lambda'}$.

\begin{proposition}\label{prop2}

Let $\triple$ be a $\params$-triple and let $p:=\text{supp}(P)$. Let $\{S_{i}\}$ be any set of $k_{P}$ discs of radius $D< \text{rad}(S_{i})\leq r-D $,  such that $p(D)\subseteq \bigcup S_{i}$. For every disc $S_{i}$, let $S'_{i}:=S_{i}(t+2D)$ be a disc concentric to $S_{i}$. Suppose $\Pi_{\cP'}=\prod_{i}\Pi_{ S'_{i}(D)}$, where $\cP'=\cP(t+2D)$. Then
\begin{align}
\left\Vert P\Pi_{\cP'}- \expval{P}{\psi}\Pi_{\cP'}\right\Vert_{\infty}\leq 2\delta \, k_{P}
\end{align}
\end{proposition}

\begin{proof}
The proof follows from breaking up the support of the LMP covering the support of $P$ into $k_{P}$ locally invisible regions, such that  we can apply Corollary \ref{corollary1} and then leverage Lemma \ref{fact5} to prove our proposition. Let $\Pi_{\cP}:=\prod_{i} \Pi_{S_{i}(D)}$. Then 
\begin{align}
\inftynorm{\Pi_{\cP} P \Pi_{\cP'}-P \Pi_{\cP'}}= \inftynorm{ \prod_{i=1}^{k_{P}}\Pi_{S_{i}(D)} P \prod_{i=1}^{k_{P}} \Pi_{S'_{i}(D)}-P \prod_{i=1}^{k_{P}} \Pi_{S'_{i}(D)}} \leq 2\delta k_{P},
\end{align}
where we have made use of Corollary \ref{corollary1} in the last step. 
Using Facts \ref{fact:LMP_stab} and \ref{fact5} respectively it follows that $\Pi_{\cP} P \Pi_{\cP'}=\Pi_{\cP} P \Pi_{\cP} \Pi_{\cP'}=\expval{P}{\psi}\Pi_{\cP}  \Pi_{\cP'}=\expval{P}{\psi} \Pi_{\cP'}$ and therefore 
\begin{align}
\inftynorm{\Pi_{\cP} P \Pi_{\cP'}-P \Pi_{\cP'}}=\inftynorm{ \expval{P}{\psi}\Pi_{\cP'}-P \Pi_{\cP'}},
\end{align}
from which follows our claim.
\end{proof}

\begin{figure}[ht]
\centering
\begin{subfigure}{.5\textwidth}
  \centering
  \begin{tikzpicture}[scale=0.8]
\draw[ultra thick] (0,0) rectangle (8,8);
\draw[draw=white,ultra thick,fill=cyan, opacity=0.2] (4,4) circle (2.2cm);
\draw[draw=white,ultra thick,fill=white] (4,4) circle (1.2cm);
\draw[red,ultra thick] (4,4) circle (1.8);

\node at (7,0.5) {$\Lambda$};
\node at (3.5,3.5) {$\textcolor{red}{P}$};
\node at (2.2,2.2) {$\textcolor{blue}{\Pi_{\cP'}}$};
\end{tikzpicture}
\caption{} 
\label{fig:prop_2_setup}
\end{subfigure}%
\begin{subfigure}{.5\textwidth}
  \centering
\begin{tikzpicture}[scale=0.8]
\draw[dashed, thick] (0,4) -- (8,4);
\draw[ultra thick] (0,0) rectangle (8,8);
\draw[red,ultra thick] (3,4) circle (2 cm);
\draw[blue,ultra thick] (5,4) circle (2cm);
\draw[<->,thick] (4,2.7) -- (4,5.3);
\node at (4.5,3.6) {$d_{\text{sep}}$};
\node at (6,4.3) {$M^{c}$};
\node at (5.9,3.7) {$M$};
\node at (7,0.5) {$\Lambda$};
\node at (3,6.3) {$\textcolor{red}{P_{}}$};
\draw [draw=white, fill=cyan, opacity=0.2] (4,2.3) circle (0.85 cm);
\draw [draw=white, fill=violet, opacity=0.2] (4,5.7) circle (0.85 cm);
\draw [draw=white, fill=green, opacity=0.2] (2,5.5) circle (0.85 cm);
\node at (5,6.3) {$\textcolor{blue}{Q_{}}$};
\end{tikzpicture}
\caption{} \label{fig:setup}
  \label{fig:operators}
\end{subfigure}
\caption{
(a) An example of a LMP $\Pi_{\cP'}$ (cyan) covering the support of an approximately locally invisible loop operator $P$ (red). (b) A pair of intersecting $(\delta,A,B)$-locally invisible loop operators $P$ and $Q$ supported on the lattice $\Lambda$ partitioned into two regions $M$ and $M^c$, such that at most one region of intersection of the loop operators is contained in either half. Note that the partition may be  chosen arbitrarily as long as it satisfies this criterion. Let $\cP'$ be a thickening of $\text{supp}(P)$ as shown in Figure \ref{fig:prop_2_setup} and defined in Proposition \ref{prop3}. We may write $\cP'$ as the union of the three sets $V_{1}:=\{S_{i}(D)
\subset \cP': S_{i}(D)\subset M \}$ (cyan shade)  and $V_2:=\{S_{i}(D)\subset \cP': S_{i}(D)\cap \text{supp}(Q)=\emptyset \}$ (green shade) and $V_3=\{S_{i}(D)\subset \cP': (S_{i}(D)\subset M^c)\cap (S_{i}(D) \cap \text{supp}(Q)\neq\emptyset) \}$ (violet shade).} 
\end{figure}
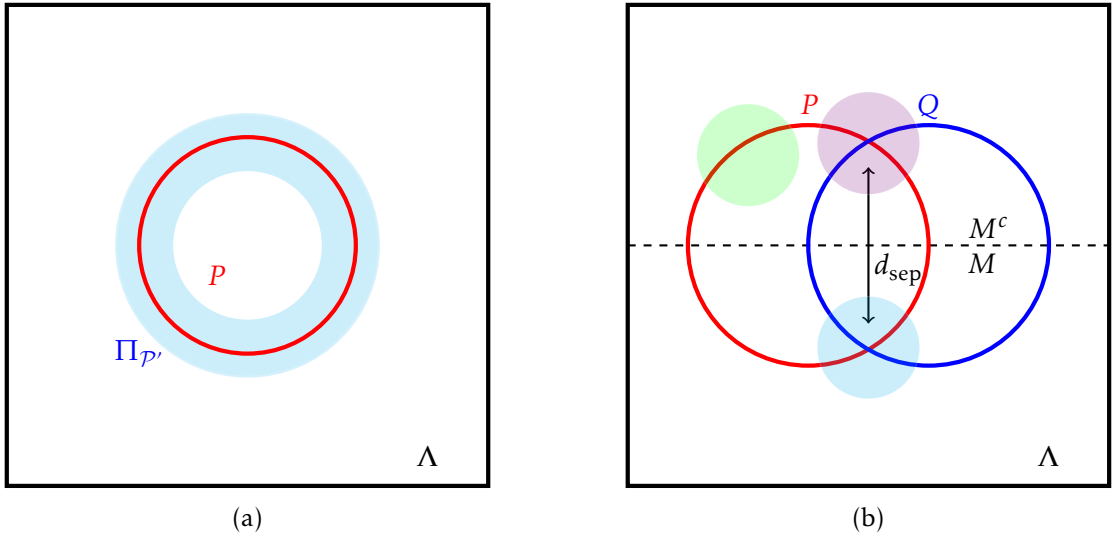

\section{The twist product }\label{sec:twist_product}
We will now introduce a key concept behind our circuit lower bound proof technique. This is the notion of the twist product. 

\begin{definition}(Twist product)
Let $P$ and $Q$ be two bipartite operators that w.l.o.g admit decompositions $P=\sum_{j} P_{M}^{j}\otimes P_{M^{c}}^{j}$ and $Q=\sum_{k}Q_{M}^{k}\otimes Q_{M^{c}}^{k}$.
Then we define their \textit{twist product} as
\begin{align}
P\infty Q:=\sum_{j,k} P_{M^c}^{j}Q_{M^c}^{k}\otimes Q_{M}^{k}P_{M}^{j}
\end{align} 
\end{definition}
\noindent By summing over either of the two indices one may also conveniently write this as:
\begin{align}\label{eq:twist_decomp}
P\infty Q=\sum_{j} P_{M^c}^{j} Q P_{M}^{j}
=\sum_{k} Q_{M}^{k} P Q_{M^c}^{k}
\end{align}
This product was first introduced in \cite{Haah_2016} and can be used to construct a witness function for the presence of long-range entanglement in a state $\ket{\psi}$. Haah shows how one can do so by finding two non-trivial exactly locally invisible operators $P$ and $Q$ with respect to $\ket{\psi}$ that intersect at two points.
He proves that if $\expval{P\infty Q}{\psi}\neq\expval{P }{\psi}\expval{Q}{\psi}$, then the state $\ket{\psi}$ is not connected to a product state via a geometrically local circuit of constant depth and therefore must exhibit long-range entanglement.

As a concrete example, Haah considers the toric code Hamiltonian defined on a sphere. 
For this, he identifies $P$ and $Q$ with an $X$ and $Z$ type stabilizer of the toric code respectively. Both are exactly locally invisible by definition and form closed loops that can be put on the lattice such that they intersect at two points (see Figure \ref{fig:setup}). Their so-called 'twist paring' $\expval{P\infty Q}{\psi}$ evaluates to (-1) and may be understood as the phase acquired by the wavefunction when braiding an $e$ and an $m$ excitation in the toric code around each other. Therefore $\expval{P\infty Q}{\psi}\neq\expval{P }{\psi}\expval{Q}{\psi}$, which implies that the ground state of the toric code on a sphere cannot be prepared with a constant depth circuit. Such states are sometimes referred to as being \textit{topologically ordered}.

Developing this line of argument further then allows for a rigorous proof of a circuit lower bound when the ground state is unique. This setting poses an obstacle to previously established techniques for circuit depth lower bounds \cite{Bravyi_hastings_verstraete,Koenig_Pastawski}, which only deal with the case where the ground space is degenerate.

In this section, we will extend this framework and show that it may still be used when $P$ and $Q$ are only $\params$-approximately locally invisible with respect to $\ket{\psi}$. One may therefore say that the witness function is robust to perturbations to $\ket{\psi}$ in purified distance.

The main purpose of this section is to prove an upper bound on this robust witness function 
\begin{align}
C(P,Q)_{\ket{\psi}}:=\left|\expval{P\infty Q }{\psi}- \expval{P }{\psi} \expval{Q}{\psi} \right| 
\end{align} in terms of the circuit depth $D$ required to prepare $\ket{\psi}$ (see Theorem \ref{thm2}). Later, we will combine this with a lower bound on this function in terms of the local invisibility parameter $\delta$ in order to give a circuit lower bound for low-energy states. 

To prove Theorem \ref{thm2}, we first have to establish some properties of the twist product and how it can be combined with LMPs for our purposes. A key property of the twist product is the following:

\begin{lemma}\label{lem:clean}
Consider a bipartite lattice $\Lambda= M \cup M^{c}$ and let $P,Q,O\in \mathcal{B}(\mathcal{H}_{\Lambda})$. If $\text{supp}(O)\subset M$ or $\text{supp}(O)\cap \text{supp}(Q)=\emptyset $, then $(PO)\infty Q=(P\infty Q) O$.
\end{lemma}

\begin{proof}
This is immediate from rewriting the twist product as in Eq. (\ref{eq:twist_decomp}).
If $\text{supp}(O)\cap \text{supp}(Q)=\emptyset $
then we have $(P O)\infty Q=\sum_{j}Q_{M}POQ_{M^c}=\sum_{j}Q_{M}PQ_{M^c}O=(P\infty Q) O$.
\noindent If $\text{supp}(O)\subset M$, we have 
$(P O)\infty Q=\sum_{j} P_{M^c}^{j} Q P_{M}^{j}O=(P \infty Q)O$
\end{proof}

\noindent From here onwards, we will specifically consider the twist product of two operators $P,Q\in\mathcal{B}(\mathcal{H}_{\Lambda})$ supported on two intersecting annuli  (see Fig. \ref{fig:operators}) and our choice of bipartition $\Lambda=M \cup M^c$ will separate the two points at which the annuli intersect. Note however that our proof techniques could straightforwardly be generalized to settings where both $P$ and $Q$ do not have annular support. The only requirement is for these operators is to intersect at two points that are sufficiently far apart from each other.
\\

The following proposition will be instrumental in proving Theorem \ref{thm2}:

\begin {proposition}\label{prop3}
Let $\left(\{P,Q\},\,\,\rho_{\Lambda},\,\Pi\right)$ be a $\params$-triple, where $P$ and $Q$ are constant-depth unitaries that are supported on 
intersecting annuli (see Fig. \ref{fig:operators}) of equal diameter $d_{P}$. Let $p:=\text{supp}(P)$ denote the annular support and assume that it has a thickness of $\tau_{p}=\bigO(1)$. 

Let $\{S_{i}\}$ be a set of 
discs such that $ \bigcup_{i} S_{i}\supseteq p(t+3D)$, where $3D+t+\tau_{p} \leq \text{rad}(S_{i}) \leq r-D$.
For every disc $S_{i}$, let $S'_{i}:=S_{i}(t+2D)$ be a disc concentric to $S_{i}$ such that $\bigcup_{i} S'_{i}\supseteq p\left(2t+5D\right)$.
Suppose that $\Pi_{\cP'}:=\prod_{i} \Pi_{S_{i}(D)}$ and $\Pi_{\cP''}=\prod_{i}\Pi_{ S'_{i}(D)}$, where $\cP'=\cP(t+2D)$ and 
let $k_{Q}$ be an upper bound on the number of discs $S_{i}(D)$ that touch one of the intersections of $P$ with $Q$. Further assume that the intersections of both annuli are apart by a distance $d_{\text{sep}} > 2\, \text{rad}\left(S'_{i}(D)\right)$.  Then 
\begin{align}
    \inftynorm{ P\infty Q \ket{\psi}_{\Lambda'} - P\Pi_{\cP'}\infty Q \ket{\psi}_{\Lambda'} }\leq 2\delta \,k_{Q} \, \alpha_{P} 
\end{align}
where $\alpha_{P}:=\sum_{j} \left\Vert P_{M}^{j}\right\Vert_{\infty}\left\Vert P_{M^c}^{j}\right\Vert_{\infty}$ for $P=\sum_{j}P_{M}^{j}\otimes P_{M^c}^{j}$.
\end{proposition}
\begin{proof}
We will prove this proposition via Lemma \ref{lem:clean}, Proposition \ref{prop2} and Corollary \ref{corollary1}. 
We shall consider decomposing $\cP'=V_{1}\cup V_{2}\cup V_{3}$ into three sets, that each correspond to the three distinct cases illustrated in Figure \ref{fig:setup}. Accordingly, we will consider the projection $\Pi_{\cP'}$, and decompose it with respect to this partitioning. We then write
\begin{align}
\Pi_{\cP'}=\prod_{S_{i}\in V_3}\Pi_{S_{i}(D)}\prod_{S_{i}\in V_{1}\cup V_{2}} \Pi_{S_{i}(D)}
\end{align}
where, with slight abuse of notation, $S_{i}\in V_{m}$ and later $S'_{i}\in V_{m}$ for $m\in\{1,2,3\}$, if  $S_{i}(D)\in V_{m}$.  
 Using Lemma \ref{lem:clean} and Fact \ref{fact:LMP_stab} we then have
\begin{align}
&P\Pi_{\cP'}\infty Q\ket{\psi}_{\Lambda'}=\left(P\prod_{S_{i}\in V_3}\Pi_{S_{i}(D)}\infty Q\right) \prod_{S_{i}\in V_{1}\cup V_{2}}\Pi_{S_{i}(D)}\ket{\psi}_{\Lambda'}=\left(P\prod_{S_{i}\in V_3}\Pi_{S_{i}(D)}\infty Q\right)\ket{\psi}_{\Lambda'}\\
&=\left(P\prod_{S_{i}\in V_3}\Pi_{S_{i}(D)}\infty Q\right)\prod_{S'_{i}\in V_3}\Pi_{S'_{i}(D)}\ket{\psi}_{\Lambda}.
\intertext{Therefore} 
&\left\Vert \left(P\Pi_{\cP'}\infty Q\right) \ket{\psi}_{\Lambda'}- P\infty Q\ket{\psi}_{\Lambda'}\right\Vert= \left\Vert \left(P\prod_{S_{i}\in V_{3}}\Pi_{S_{i}(D)}\infty Q\right) \ket{\psi}_{\Lambda'}- P\infty Q\ket{\psi}_{\Lambda'}\right\Vert
\intertext{and using Fact \ref{fact:LMP_stab} once again, we have }
&=\left\Vert \left(P\prod_{S_{i}\in V_{3}}\Pi_{S_{i}(D)}\infty Q\right)\prod_{S'_{i}\in V_{3}}\Pi_{S'_{i}(D)} \ket{\psi}_{\Lambda'} - (P\infty Q)\prod_{S'_{i}\in V_{3}}\Pi_{S'_{i}(D)}  \ket{\psi}_{\Lambda'}\right\Vert.
\intertext{We will now choose a partitioning $\Lambda=M\cup M^c$ such that the cut is chosen to be as far away from $\text{supp}(P)\cap\text{supp}(Q)\cap M^c$ as possible, but still separates it from $\text{supp}(P)\cap\text{supp}(Q)\cap M$. From  Eq.(\ref{eq:twist_decomp}) we get}
&=\left\Vert \sum_{j}P_{M^c}^{j}\prod_{S_{i}\in V_{3}}\Pi_{S_{i}(D)} Q P_{M}^{j}\prod_{S'_{i}\in V_{3}}\Pi_{S'_{i}(D)} \ket{\psi}_{\Lambda'}-\left(\sum_{j} P_{M^c}^{j} Q P_{M}^{j}\prod_{S'_{i}\in V_{3}}\Pi_{S'_{i}(D)}\right)      \ket{\psi}_{\Lambda'}\right\Vert\\
&\leq \sum_{j} \inftynorm{P_{M^c}^{j}\prod_{S_{i}\in V_{3}}\Pi_{S_{i}(D)} Q P_{M}^{j}\prod_{S'_{i}\in V_{3}}\Pi_{S'_{i}(D)}-P_{M^c}^{j} Q P_{M}^{j}\prod_{S'_{i}\in V_{3}}\Pi_{S'_{i}(D)} },\intertext{where the upper bound follows from definition of the operator norm and triangle inequality. Since $d_{\text{sep}} > 2\, \text{rad}\left(S'_{i}(D)\right)$, our choice of partitioning implies that $S'_{i}(D)\subset M^c $ if $S'_{i}\in V_3$. Therefore $\left[\prod_{S'_{i}\in V_{3}}\Pi_{S'_{i}(D)}, \,P_{M}^{j}\right]=0$ for all $j$, such that:}
&=\sum_{j} \inftynorm{P_{M^c}^{j}\prod_{S_{i}\in V_{3}}\Pi_{S_{i}(D)} Q \prod_{S'_{i}\in V_{3}}\Pi_{S'_{i}(D)} P_{M}^{j}-P_{M^c}^{j} Q \prod_{S'_{i}\in V_{3}}\Pi_{S'_{i}(D)} P_{M}^{j} }\\
&\leq \left(\sum_{j}\inftynorm{P_{M^c}^{j}}\inftynorm{P_{M}^{j}}\right)\inftynorm{\prod_{S_{i}\in V_{3}}\Pi_{S_{i}(D)} Q \prod_{S_{i}\in V_{3}}\Pi_{S'_{i}(D)}-Q \prod_{S'_{i}\in V_{3}}\Pi_{S'_{i}(D)}}
\leq 2\delta \,k_{Q}\,\alpha_{P} ,
\end{align}
where the last step follows from applying Corollary \ref{corollary1} and assuming that at most $k_{Q}$ discs from the set $\{S_{i}\}$ overlap with with $\text{supp}(Q)\cup M $.
\end{proof}

Next we prove an upper bound on our witness function 
\begin{align}
C(P,Q)_{\ket{\psi}}:=\left| \expval{P\infty Q }{\psi}- \expval{P }{\psi} \expval{Q}{\psi} \right|,
\end{align}
which can be understood as a ``robust'' version of Haah's witness function \cite{Haah_2016}. 

Upper bounding this approximate witness function will be instrumental in proving our circuit-depth lower bound. This will follow with the help of Proposition \ref{prop2} and \ref{prop3}. Our setup will be the same as in Figure \ref{fig:setup}. 

\begin{theorem}\label{thm2}
Let $\left(\{P,Q\}, \Pi, \rho_{\Lambda}\right)$ be a $\params$-triple, where $P$ and $Q$ are as defined in Proposition $\ref{prop3}$. 
Then for $k_{P}$ as defined in Proposition \ref{prop2} and $k_{Q}$ as defined in Proposition \ref{prop3}, we have that 
\begin{align}\label{eq:twist_corr_upper_bound}
\left| \expval{P\infty Q }{\psi}- \expval{P }{\psi} \expval{Q}{\psi} \right|\leq 2\delta \left(\alpha_{P}k_{Q}+\alpha_{Q}k_{P}\right),    
\end{align}
with $\alpha_{P}$ and $\alpha_{Q}$ being defined as in Proposition \ref{prop3}.
\end{theorem}

\begin{proof}
The proof follows via triangle inequality, as well as using Eq. (\ref{eq:twist_decomp}) and our bounds proved in Propositions \ref{prop2} and \ref{prop3} respectively.\\
Let $\cP'\supset p(D)$ be as defined in Proposition \ref{prop3}.
Adding and substracting the two terms $\expval{P\Pi_{\cP'}\infty Q}{\psi}$ and $ \expval{\Pi_{\cP'}\infty Q}{\psi}\expval{P}{\psi}$ to the LHS of Eq. (\ref{eq:twist_corr_upper_bound}) and applying triangle inequality, we get: 
\begin{align}
\left| \expval{P\infty Q }{\psi}- \expval{P}{\psi} \expval{Q}{\psi} \right|\nonumber 
&\leq \left|\expval{P\infty Q}{\psi}-\expval{P\Pi_{\cP'}\infty Q}{\psi}\right| \nonumber\\
& + \left|\expval{P\Pi_{\cP'}\infty Q}{\psi}-\expval{\Pi_{\cP'}\infty Q}{\psi}\expval{P}{\psi}\right|\nonumber\\
&+\left|\expval{\Pi_{\cP'}\infty Q}{\psi}\expval{P}{\psi}-\expval{ Q}{\psi}\expval{P}{\psi}\right|,\label{eq:corr_triangle_bound}\\
&\leq 2\delta\, k_{Q}\alpha_{P} +  \left|\expval{\left(P\Pi_{\cP'}-\Pi_{\cP'}\expval{P}{\psi}\right)\infty Q}{\psi}\right|,
\end{align}
where the third term vanishes since $\expval{\Pi_{\cP'}\infty Q}{\psi}=\expval{Q}{\psi}$ and where the first term can be upper bounded using Proposition \ref{prop3}. Using Cauchy-Schwarz on the second term and switching to operator norm it follows that 
\begin{align}
\left| \expval{P\infty Q }{\psi}- \expval{P}{\psi} \expval{Q}{\psi} \right|&\leq 2\delta\,k_{Q}\alpha_{P} + \inftynorm{P\Pi_{\cP'}-\Pi_{\cP'}\expval{P}{\psi}}\sum_{j}\inftynorm{Q_{M^c}^j}\inftynorm{Q_{M}^j}\intertext{where we have used that $\inftynorm{P\infty Q}= \inftynorm{\sum_{j} Q_{M}^{j} P Q_{M^c}^{j}}\leq \inftynorm{P} \sum_{j} \inftynorm{Q_{M}^{j}}  \inftynorm{Q_{M^c}^{j}} $. Bounding the second term using Proposition \ref{prop2}, we obtain }
&\leq 2\delta k_{Q}\alpha_{P} +2\delta k_{P}\alpha_{Q}=2\delta \left(\alpha_{P}k_{Q}+\alpha_{Q}k_{P}\right)
\end{align}
\end{proof}

\section{Circuit depth lower bound for low-energy states}\label{sec:low_energy_states}

We will now apply the tools we have developed in the previous sections in order to prove a circuit depth lower bound for low-energy states $\rho^{\epsilon}_{\Lambda}$ of of the toric code Hamiltonian in 2D with no assumptions on the ground space degeneracy.

We will prove our result by showing that the stabilizers of toric code  yield approximately locally invisible operators for low-energy states. Our starting point will be to show that for any such state $\rho^{\epsilon}_{\Lambda}$, one can find a region $R
\subset \Lambda$ on which the state has high overlap with its projection onto the ground space of $H_{R}$, i.e. the restriction of $H_{\Lambda}$ to region $R$.
This implies that any toric code stabilizer $s_{R}$ supported on region $R$ that is exactly locally invisible with respect to any of the states in the ground space of $H_{R}$, is approximately locally invisible with respect to low-energy states $\rho^{\epsilon}_{\Lambda}$ of $H_{\Lambda}$. 

While we consider the toric code for concreteness, our bound applies to all commuting projector models in 2D for which there exist local operators that stabilize any ground state and that are constant depth unitaries.

\subsection{Approximately locally invisible operators for low-energy states}
\label{sec:good_subsystem} 

 We now introduce the notion of \emph{good subsystem} and show how it can be used to find approximately locally invisible operators for low-energy states of certain classes of Hamiltonians. Intuitively, a ``good'' subsystem is a subsystem whose energy density within that subsystem is ``sufficiently low.'' 
 This Section contains two elementary but useful facts. The first fact is that, if a global state has an energy density of $\epsilon$, for any $\ell \leq L$, there is a subsystem of size $\ell \times \ell$ with local energy density at most $\epsilon$. We shall refer to such a subsystem as a $\epsilon$-good subsystem. The second fact is that an $\epsilon$-good subsystem has $\Omega(1)$ overlap with some state in the ground space of the subsystem, provided that $\epsilon$ is smaller than the inverse of the subsystem size.

\noindent Let us state the formal definition below.
\begin{definition}
A subsystem $R\subset \Lambda$ is $\epsilon$-good with respect to $\rho_{\Lambda}$ and a Hamiltonian $H_{\Lambda}$ if $\rho_{\Lambda}$ has an energy density of at most $\epsilon$ over $R$.
\end{definition} 


\begin{lemma}
\label{lemma:existence_good}
Let $\rho_{\Lambda}^{\epsilon}$ be a state with energy density $\epsilon$, i.e. $\Tr[H_{\Lambda}\rho^{\epsilon}]=\epsilon\abs{\Lambda}$. For any $\ell$ such that $2\leq \ell \leq L$, there is a subsystem $R$ of size $\ell \times \ell$ such that $R$ is $\epsilon$-good.
\end{lemma}
\begin{proof}
Let $v_{\ell}$ be a $\ell \times \ell$ subsystem such that its bottom left corner is $v\in \Lambda$. We may write our Hamiltonian as:
\begin{equation}
    H_{\Lambda} = \sum_{v\in \Lambda} \frac{H_{v_{\ell}}}{\ell^2}.
\end{equation}
Thus, we can conclude that
\begin{equation}
    \text{Tr}\left[\rho_{\Lambda}^{\epsilon} \frac{H_{\Lambda}}{L^2}\right] = \mathbb{E}_{v\sim \Lambda} \left[\text{Tr}\left[ \rho_{\Lambda}^{\epsilon} \frac{H_{v_{\ell}}}{\ell^2} \right] \right], \label{eq:temp646}
\end{equation}
where $\mathbb{E}_{v\sim \Lambda}[\cdot]$ is the average over $v$ chosen uniformly randomly over $\Lambda$. Since the left hand side of Eq.~\eqref{eq:temp646} is $\epsilon$, we conclude that the mean of $\text{Tr}\left[ \rho_{\Lambda} \frac{H_{v_{\ell}}}{\ell^2} \right]$ ought to be $\epsilon$. If this quantity is strictly larger than $\epsilon$, we arrive at a contradiction. Therefore, there exists a $v\in \Lambda$ such that $    \text{Tr}\left[ \rho_{\Lambda}^{\epsilon} \frac{H_{v_{\ell}}}{\ell^2} \right] \leq \epsilon.$
\begin{equation}
    \text{Tr}\left[ \rho_{\Lambda}^{\epsilon} \frac{H_{v_{\ell}}}{\ell^2} \right] \leq \epsilon.
\end{equation}
\end{proof}
Having shown that for any low-energy  state $\rho^{\epsilon}$ there exists a local patch $R$ with energy density at most $ \epsilon$, we are now in position to show that any low-energy state is close to its projection onto the ground space of $H_{R}$ provided certain conditions on $H_{\Lambda}$ hold. While we can prove such a statement for general frustration-free Hamiltonians with a local gap (see Appendix \ref{app: ff_leakage_bound}), we will only consider frustration-free, commuting Hamiltonians for now.

We show how these conditions imply that the projection onto the ground space of $H_{R}$ approximately stabilizes the low-energy state, i.e., $\rho_\Lambda^{\epsilon}\approx \Pi_R^{\text{0}} \rho_\Lambda^{\epsilon} \Pi_R^{\text{0}}$.
\begin{proposition}\label{prop:leakage}
Suppose $R$ is $\epsilon$-good with respect to $\rho_{\Lambda}^{\epsilon}$. Then 
\begin{equation}
\purifieddistance \left(\rho_\Lambda^{\epsilon}, \frac{\Pi_{R}^{0} \rho_\Lambda^{\epsilon} \Pi_{R}^{0}}{\text{Tr}\left[\Pi_{R}^{0} \rho_\Lambda^{\epsilon} \Pi_{R}^{0}\right]} \right) \leq \sqrt{ |R|\epsilon}
  \end{equation}

\end{proposition}
The proof of Proposition~\ref{prop:leakage} is based on two well-known facts: the quantum union bound \cite{Khabbazi_Oskouei_2019} and the gentle measurement lemma \cite{Wilde}. We state these facts below for the readers' convenience. 

\begin{fact}[Commutative union bound\cite{Khabbazi_Oskouei_2019}]\label{lem:union_bound}
Let $\rho \in \density$ and $\{\Pi_{i}\}$ be a set of commuting projectors. Then we have 
\begin{equation}
    1-\Tr[\left(\prod_{i}\Pi_{i}\right) \rho]\leq \sum_{i} \Tr[(\mathbbm{1}-\Pi_{i})\rho]. \label{eq:union_bound}
\end{equation}
\end{fact}

\begin{fact}[Gentle measurement lemma~\cite{Wilde}]\label{lem:gentle_meas}
Let $\Pi$ be a projection and $\rho\in \density$, such that $\Tr[\Pi\rho]\geq 1-\epsilon$. For $\rho'=\frac{\Pi\rho \Pi}{\Tr[\Pi\rho \Pi]}$, the following bound holds:
\begin{equation}
    \purifieddistance(\rho, \rho')\leq \sqrt{\epsilon}.
\end{equation}
\end{fact}

The gentle measurement lemma is usually stated in terms of the trace distance. Here we use a version that is more suitable for our purpose and is also tighter by a constant factor. This version of the lemma follows from Eq.(9.202) in Ref.~\cite{Wilde}.

Now we are in a position to prove Proposition~\ref{prop:leakage}.

\emph{Proof of Proposition~\ref{prop:leakage}.} Using the union bound (Fact~\ref{lem:union_bound}) to lower bound the energy of $\rho_{\Lambda}^{\epsilon}$ over the patch $R$, we have 
\begin{equation}
\begin{aligned}
    \text{Tr} \left[ \rho^{\epsilon}_\Lambda H_{R}\right]\geq 1-\Tr[\left(\prod_{s\in S_R}\frac{\mathbbm{1}+s}{2}\right)\rho^{\epsilon}_{\Lambda}].
    \end{aligned}
\end{equation}
As $R$ is $\epsilon$-good it follows that
\begin{equation}
|R| \epsilon \geq 1-\Tr[\left(\prod_{s\in S_R}\frac{\mathbbm{1}+s}{2}\right) \rho^{\epsilon}_{\Lambda}].
\end{equation}
Since $\prod_{s\in S_R} \frac{I+s}{2}=\Pi_{R}^{0}$, we conclude
\begin{equation}
    \Tr[\Pi_{R}^{0} \rho^{\epsilon}_{\Lambda}]\geq 1-|R|\epsilon.
\end{equation}
From Fact~\ref{lem:gentle_meas}, the claim follows.
\[
\pushQED{\qed} 
\qedhere
\popQED
\]     
This bound tells us that distance of $\rho^{\epsilon}_{R}$ to the local ground state  $\rho^{0}_{R}:=\frac{\Pi_{R}^{0} \rho_{R}^{\epsilon} \Pi_{R}^{0}}{\text{Tr}\left[\Pi_{R}^{0} \rho_{R}^{\epsilon} \Pi_{R}^{0}\right]}$ is constant, if the energy of $\rho_{R}^{\epsilon}$ is a constant. Therefore, a low-energy state $\rho^{\epsilon}_{\Lambda}$ actually provides a good approximation to a local ground state on any region $R$ with a constant energy $\abs{R}\epsilon<1$.

The bound given in Proposition \ref{prop:leakage} is key in allowing us to find approximately locally invisible operators for low-energy states $\rho^{\epsilon}_{\Lambda}$ of the toric code Hamiltonian. 

\begin{corollary}\label{coroll:approx_op_low_energy}
Let $\rho^{\epsilon}_{\Lambda}$ be a state of energy density at most $\epsilon$ with respect to the toric code Hamiltonian $H_{\Lambda}$. Suppose $R$ is an $\epsilon$-good subsystem. Then any any $s\in \mathcal{S}^{\text{TC}}_R$ is $\params$-approximately locally invisible with respect to $\rho^{\epsilon}_{\Lambda}$ for $\delta = 2\sqrt{|R|\epsilon}\,\,$ and $r=\bigO\left(\sqrt{\Lambda}\right)$ and $t=2$.
\end{corollary}

\begin{proof}
Let $\rho^{0}_{\Lambda}:=\frac{\Pi_{R}^{0} \rho_{\Lambda}^{\epsilon} \Pi_{R}^{0}}{\text{Tr}\left[\Pi_{R}^{0} \rho_{R}^{\epsilon} \Pi_{R}^{0}\right]}$.
As the support of $s$ is contained in $R$ it follows from Fact \ref{fact:exact_invisibility} that any $s$ is $(r, t)$-locally invisible with respect to $\rho^{0}_{\Lambda}$ for $r=\bigO\left(\sqrt{\Lambda}\right)$ and $t=\Omega(1)$. Since $s$ is a depth-1 unitary, one may choose $t=2$.
From Proposition $\ref{prop:leakage}$ and Lemma  \ref{lem:robustness_local_invisibility} our claim then follows.
\end{proof}
This is precisely what we need in order to apply our formalism developed in Sections \ref{sec:local_invisibility}-\ref{sec:twist_product} to low-energy states. This is the content of our main theorem.

\subsection{Main theorem}

Having found non-trivial $\params$-approximately locally invisible operators for low-energy states $\rho^{\epsilon}_{\Lambda}$ of the toric code Hamiltonian, we are now in a position to prove a circuit-depth lower bound for these states. 
\main*


\begin{proof}

Our proof follows from giving upper and lower bounds on $C(P,Q)_{\ket{\psi}}$ for $P$ and $Q$ chosen to be toric code stabilizers supported on the $\epsilon$-good subsystem $R$ of $\rho^{\epsilon}_{\Lambda}$. We will prove an upper bound in terms of the circuit depth $D$ required to prepare any low-energy state $\rho^{\epsilon}_{\Lambda}$ and a lower bound in terms of the energy of $\rho^{\epsilon}_{\Lambda}$ over the patch $R$.

Let $\ketexcited_{\Lambda'}$ be the purification of $\rho^{\epsilon}_{\Lambda}$. Let $R\subset\Lambda\subset\Lambda'$ be an $\epsilon$-good patch and let $ s_{R}^{X}, s_{R}^{Z} \in S^{\text{TC}}_{R} $, where the superscript denotes whether it is an $X$-type or $Z$-type stabilizer, which are $\params$-approximately invisible with respect to $\rho_{\Lambda}^{\epsilon}$ as shown in Corollary \ref{coroll:approx_op_low_energy}. Then
\begin{align}\label{eq:corr_lower_bound}
C\left(s_{R}^{X},s_{R}^{Z}\right)_{\ketexcited}=\abs{\expval{s_{R}^X \infty s_{R}^{Z}}{\psi^{\epsilon}}-\expval{s_{R}^{X}}{\psi^{\epsilon}}\expval{s_{R}^{Z}}{\psi^{\epsilon}}}\geq 2-3\sqrt{|R|\epsilon}+|R|\epsilon,
\end{align}
where $C\left(s_{R}^{X},s_{R}^{Z}\right)_{\ket{\psi^{0}}}=2$. This bound immediately follows from the fact that 
\begin{align}
\left|\expval{O_{R}}{\psi^{\epsilon}}-\expval{O_{R}}{\psi^{0}}\right|\leq  \left\Vert O_{R}\right\Vert_{\infty }\left\Vert\rho^{\epsilon}_{R}-\rho^{0}_{R}\right\Vert_{1},\end{align} 
for any $O_{R}\in \mathcal{B}(\mathcal{H}_{R})$, where $\ket{\psi^{0}}_{\Lambda'}$ is a purification of $\rho^{0}_{\Lambda'}$ as defined in the proof of Corollary \ref{coroll:approx_op_low_energy}. Since  $\left\Vert s^{X}_{R}\infty s^{Z}_{R}\right\Vert_{\infty}=\left\Vert s^{X}_{R}\right\Vert_{\infty}=\left\Vert s^{Z}_{R}\right\Vert_{\infty}=1$ and $\expval{s^{X}_{R}}{\psi^{\epsilon}}=\expval{s^{Z}_{R}}{\psi^{\epsilon}}=1$ and $\expval{s^{X}_{R}\infty s^{X}_{R}}{\psi^{\epsilon}}=-1 $ our lower bound in Eq.~\eqref{eq:corr_lower_bound}. From Proposition $\ref{prop:leakage}$ we have that $\left\Vert\rho^{\epsilon}_{R}-\rho^{0}_{R}\right\Vert_{1}\leq \sqrt{|R|\epsilon}$.
\\

Next, we prove an upper bound on $C\left(s_{R}^{X},s_{R}^{Z}\right)_{\ketexcited}$ via Theorem \ref{thm2}.
For our choices of $P=s^X_{R}$ and $Q=s^Z_{R}$, we have that $\delta=2\sqrt{\abs{R}\epsilon}$\, and $\alpha_{P}=\alpha_Q=1$ and therefore 
\begin{align}
C\left(s_{R}^{X},s_{R}^{Z}\right)_{\ketexcited}\leq 2\sqrt{\abs{R}\epsilon}(k_{P}+k_{Q})
\end{align}
To determine $k_{Q}$, we can choose the radius of the disc-shaped supports $S_{i}(D)$ of the LMPs in Proposition \ref{prop3} to be $\text{rad}(S_{i}(D))= 4D+t+\tau_{p}$, where $\tau_{p}$ denotes equals the thickness of the annular support of $s^x_{R}$ and $t$ denotes the local invisibility parameter.  Thus, $k_Q= 4\,\left(4D+t+\tau_{p}\right)$, since at most $4\,\,\text{rad}\left(S_{i}(D)\right)$ discs will overlap with one of the intersections of the two annuli formed by $s^X_{R}$ and $s^Z_{R}$.

For $k_{P}$ we can independently choose $\text{rad}\left(S_{i}(D)\right)=\mathcal{O}(\sqrt{R})$, such that $k_{P}=1$. This can be done, since $k_{P}$ and $k_{Q}$ bound independent terms in Eq.~\eqref{eq:corr_triangle_bound}. 
Given these choices of $k_{P}$ and $k_{Q}$, we have 
\begin{align}
C\left(s_{R}^{X},s_{R}^{Z}\right)_{\ketexcited}\leq 2\sqrt{\abs{R}\epsilon}\left(9+16\,D\,\right),
\end{align}
as $t=2$ and $\tau_{p}=1$ for the stabilizer $s^{x}_{R}$.
Putting both upper and lower bounds together,  we then have
\begin{align}
2-3\sqrt{|R|\epsilon}+|R|\epsilon\leq C\left(s_{R}^{X},s_{R}^{Z}\right)_{\ketexcited}&\leq 2\sqrt{\abs{R}\epsilon}\left(9+16\,D\,\right),\\
\frac{1}{16\sqrt{\abs{R}\epsilon}}-\frac{21}{32}+\frac{1}{32}\sqrt{\abs{R}\epsilon}&\leq D,
\end{align}
implying that $D=\Omega\left(1/\sqrt{\abs{R}\epsilon}\right)$. To ensure a non-trivial value for the local invisibility parameter we require $\sqrt{\abs{R}\epsilon}< 1/2$ where  $\epsilon\in [0,1)$ since we need that $\abs{R}>1$. This condition is satisfied for $R=\bigO(1/\epsilon^{\alpha})$ for $\alpha\in (0,1)$. 
We thus have
\begin{align} 
D=\Omega\left(1/\epsilon^{\frac{1-\alpha}{2}}\right)\,\, \,\,\,\,\,\,\text{for}\,\,\, \alpha\in (0,1) 
\end{align}\label{eq:circuit_bound}
Our result then follows by taking the minimum between the bound in Eq.~\eqref{eq:circuit_bound}and a circuit lower bound for preparing the exact ground state. 
We can obtain a lower bound for the exact ground state by invoking Corollary III.5 of Haah's work \cite{Haah_2016} and using the stabilizers as exactly locally invisible operators for the ground state. The lower bound for preparing the ground state then becomes $D=\Omega\left(d_{\text{sep}}\right)$, where $d_{\text{sep}}$ is the distance between the two points of intersection of the annuli (see Figure \ref{fig:setup}), which are now allowed to be supported on the whole lattice $\Lambda$. Since our lattice is 2D, we have that $d_{\text{sep}}=\bigO\left(\sqrt{\abs{\Lambda}}\right)$, such that $D=\Omega\left(\sqrt{\abs{\Lambda}}\right)$. 
\end{proof}

Some comments are in order. While we have stated our main theorem in terms of the toric code with arbitrary ground space dimension, our lower bound holds for any CSS stabilizer code with stabilizers that can be contained in an $\epsilon$-good subsystem $R$ and which intersect at only two points. A natural choice for this are codes with deformable and contractable loop-like stabilizers that can be made to fit onto any patch $R$, as can be found in the 2D surface or color codes \cite{Bombin_survey}. For such codes, 
one can always find depth-1 unitaries 
that can be supported on some local patch $R$ and made to intersect at two distant points.

Note that our lower bound does not depend on the explcit value of the twist pairing $\expval{s^{X}\infty s^{Z}}{\psi}$ and therefore applies to any model with constant-depth unitaries that leave any ground state invariant . Examples for this are the quantum double models defined for abelian groups and string-net models \cite{Kitaev2003,Levin_2005} and therefore also include Hamiltonians beyond the stabilizer formalism.

It is important to note that our bounds implicitly depend on the circuit depth of the approximately locally invisible unitaries $P$ and $Q$. In particular, $\alpha_{P}=\mathcal{O}\left(\exp(\text{depth}(P)) \right)$ for a general unitary and likewise for $\alpha_{Q}$. Our circuit depth lower bound technique is therefore suitable for models for which one can find operators $P$ and $Q$ of constant depth.

\section{Outlook}\label{sec:conclusion}

To summarize, we prove a rigorous circuit-depth lower bound for preparing low-energy states of a two-dimensional topologically ordered system. The main surprise is that the complexity, quantified in terms of the depth of the circuit, is lower bounded as a polynomial in $1/\epsilon$. This suggests a subtle difference between topologically ordered Hamiltonians and trivial Hamiltonians at finite temperature, and also shows a fundamental difficulty in approximating local quantities of these systems at low energy using low-depth circuits. 

Going further, we believe that our techniques can be generalized to the setting where our local Hamiltonian is defined on an arbitrary graph and the circuits that prepare its ground states or low-energy states admit all-to-all connectivity. We hope that this might pave the way to finding Hamiltonians that satisfy the NLTS theorem with no requirement on the degeneracy of the ground space.

\vspace{1em}

\textit{Note}: During the completion of this work, we
became aware of independent work by Cong et al. \cite{Cong_22}, which also derives a circuit depth lower bound for the toric code. Our results agree where they overlap.

\section*{Acknowledgements}
AT would like to thank Dom Williamson for helpful comments. IK thanks Robert Koenig for a helpful discussion. This research was supported in part by the National Science Foundation under Grant No. NSF PHY-1748958. We thank the Kavli Institute for Theoretical Physics for hospitality while part of this work was completed. AT was supported by the Sydney Quantum Academy, Sydney, NSW, Australia.

\printbibliography

\appendix

\section{Leakage bound for locally gapped frustration-free  Hamiltonians}\label{app: ff_leakage_bound}
Here we prove a generalization of our leakage bound from Section \ref{sec:good_subsystem} to non-commuting but frustration-free systems that are locally gapped. While we do not use this result for our main result we include it here since it may be of independent interest. Our proof will follow analagously, but we will make a slight modification by using Gao's union bound \cite{Gao_2015} for non-commutative projections and apply it to the so-called detectability lemma operator \cite{anshu2016},  instead of the ground space projector itself.

\noindent We will now briefly state a version of the detectability lemma and its converse which cater to our needs and whose proofs can be found in \cite{anshu2016}. The detectability lemma defines an approximate ground space projector (AGSP) for any gapped frustration-free  Hamiltonian in terms of a product of the local Hamiltonian terms in a particular staggered ordering  (see \cite{aharonov_detectability} for more details) .  We will use the detectability lemma and its converse together with the gentle measurement lemma to bound the distance of our low-energy state to the local ground space. 
We will now briefly state these lemmas without going into much detail and we will defer the proofs to their respective references. 

\begin{fact}(Detectability lemma)\label{fact:detectability}
Let $H=\sum_{i}\mathbbm{1}-\Pi_{i}$ be a frustration-free local Hamiltonian with gap $\gamma$, for a set of projections $\{\Pi_{i}\}$ and let $\Pi^{0}$ denote the projection onto the ground space of $H$. Then there exists an ordered product $\text{DL}(H):=\prod_{i}\Pi_{i}$ , s.t. for any $\ket{\psi}$  we have that
\begin{align}\label{eq:detectability_lemma}
\norm{\text{DL}(H)(\mathbbm{1}-\Pi^{0})\ket{\psi}}^{2}\leq \left(\frac{\gamma}{g^2}+1\right)^{-1}
\end{align}

\noindent where $g\in\mathbbm{N}^{+}$ is a constant for any geometrically local Hamiltonians in any dimension and $q\in \mathbbm{N}^{+}$.
\end{fact}
\noindent 
 Since $H$ is frustration-free, $\text{DL}(H)$ exactly stabilizes the ground space, i.e. $\text{DL}(H)\Pi^{0}=\Pi^{0}$ and we may therefore equivalently write Equation \ref{eq:detectability_lemma} as $\inftynorm{\Pi^{0}-\text{DL}(H)}\leq \left(\frac{\gamma}{g^2}+1\right)^{-\frac{1}{2}} $.
\noindent Moreover, one may apply Gao's union bound \cite{Gao_2015} to $\text{DL}(H)$ to in order prove a converse to the detectability lemma. 

\begin{fact}(Converse of the detectability lemma)
Let $\rho\in \mathcal{H}$ be an arbitrary state and let $\text{DL}(H)$ be the detectability lemma operator for a local Hamiltonian $H=\sum_{i}\mathbbm{1}-\Pi_{i}$ as defined in Fact \ref{fact:detectability}. 
\begin{align}
\Tr[\text{DL}^{\dagger}(H)\text{DL}(H)\rho]\geq 1-4\Tr[H\rho]
\end{align}
\end{fact}

The converse will be useful to us as it will allow us to make a statement about the actual distance of $\rho^{\epsilon}_{R}$ to the ground space of $H_{R}$ only in terms of $\Tr[H_{R}\rho^{\epsilon}_{R}]$ and the local gap $\gamma_{R}$.

\begin{proposition}
Let  $H_{\Lambda}=\sum_{X\subset\Lambda}\mathbbm{1}-\Pi_{X}$ be a frustration-free, non-commuting local Hamiltonian and 
suppose $R\subset \Lambda$ is $\epsilon$-good for a low-energy state $\rho_{\Lambda}^{\epsilon}$. Suppose $\Pi^{0}_{R}$ is the projection onto the ground space of $H_{R}$. Then 
\begin{align}
\purifieddistance \left(\rho_\Lambda^{\epsilon}, \frac{\Pi_{R}^{0} \rho_\Lambda^{\epsilon} \Pi_{R}^{0}}{\text{Tr}\left[\Pi_{R}^{0} \rho_\Lambda^{\epsilon} \Pi_{R}^{0}\right]} \right)&\leq \sqrt{4\abs{R}\epsilon+\left(\frac{\gamma_{R}}{g^2}+1\right)^{-1}},
\end{align}
where $\gamma_{R}$ denotes the local gap of $H_{\Lambda}$, which we assume to be a constant.

\end{proposition}

\begin{proof} 
Our proof strategy is to give a lower bound to $\Tr[\Pi^{0}_{R}\rho_{R}^{\epsilon}]$ by using the fact that we may use the converse of the detectability lemma and our low-energy assumption to give a lower bound to $\Tr[\text{DL}^{\dagger}(H_{R})\text{DL}(H_{R})\rho_{R}^{\epsilon}]$ and the fact that $\text{DL}(H_{R})$ is an approximation to $\Pi^{0}_{R}$. Our estimate for $\Tr[\Pi^{0}_{R}\rho_{R}^{\epsilon}]$ then allows to invoke the gentle measurement lemma. 
We thus have
\begin{align}
\inftynorm{\text{DL}^{\dagger}(H_{R})\text{DL}(H_{R})-\Pi^{0}_{R}}&=\sup_{\norm{\ket{\psi}}=1} \norm{\text{DL}^{\dagger}(H_{R})\text{DL}(H_{R})(\mathbbm{1}-\Pi_{R}^{0})\ket{\psi}}\\
&\leq\sup_{\norm{\ket{\psi}}=1}\inftynorm{\text{DL}^{\dagger}(H_{R})}\norm{(\text{DL}(H_{R})(\mathbbm{1}-\Pi_{R}^{0})\ket{\psi}}
\leq \left(\frac{\gamma_{R}}{g^2}+1\right)^{-1},
\end{align}
where the last line follows from the definition of the operator norm, the fact that $\inftynorm{\text{DL}^{\dagger}(H_{R})}\leq 1$ since it is a product of projections and the use of the detectability lemma. With regards to an estimate of $\Tr[\Pi^{0}_{R}\rho^{\epsilon}_{\Lambda}]$, we therefore have 
\begin{align}
\left|\Tr[\Pi^{0}_{R}\rho^{\epsilon}_{\Lambda}]-\Tr[\text{DL}^{\dagger}(H_R)\text{DL}(H_R)\rho^{\epsilon}_{\Lambda}]\right|\leq \left(\frac{\gamma_{R}}{g^2}+1\right)^{-1}.
\end{align}
Using our low-energy assumption and the converse of the detectability lemma gives

\begin{align}
\Tr[\Pi^{0}_{R}\rho^{\epsilon}_{\Lambda}]\geq 1-\left(4\abs{R}\epsilon+\left(\frac{\gamma_{R}}{g^2}+1\right)^{-1}\right).
\end{align}
From the gentle measurement lemma, our proposition then follows.\\

\end{proof}

\end{document}